\documentclass[11pt]{article} 
\usepackage[utf8]{inputenc}

\usepackage[dvipsnames]{xcolor}
\usepackage{tocloft}
\usepackage{amsmath}
\usepackage{graphicx}
\usepackage{amsthm}
\usepackage{comment}
\usepackage{amsfonts}
\usepackage{enumitem}
\usepackage{parskip}
\usepackage{geometry}
\usepackage{tikz-cd} 
\usepackage{braket}
\usepackage[width=.9\textwidth]{caption}
\usepackage{hyperref}

\usepackage[page]{appendix}
\usepackage{etoolbox}

\usepackage{fancyhdr}
\fancypagestyle{title}{            
    \fancyhf{}              
    
    \fancyhead[R]{MIT-CTP/5451}
    \addtolength{\headheight}{-1 cm}

}

\newtheoremstyle{thmstyle}
  {.5cm} 
  {.5cm} 
  {\it} 
  {} 
  {\bfseries} 
  {.} 
  {.5em} 
  {} 

\theoremstyle{thmstyle}
\newtheorem{theorem}{Theorem}[section]
\newtheorem{lemma}[theorem]{Lemma}

\newtheorem{definition}[theorem]{Definition}

\newcommand{\norm}[1]{\left\lVert#1\right\rVert}

\newcommand\Tr{\mathrm{Tr}}

\newcommand\rhobetatree{\rho_{\beta,\mathrm{tree}}}
\newcommand\rhowtree{\rho_{w,\mathrm{tree}}}
\newcommand\rhobetaij{\mathrm{Tr}_{G\backslash T_{\langle i,j\rangle}}\left(\rho_{\beta}\right)}
\DeclareMathOperator*{\Ex}{\mathbb{E}}
\DeclareMathOperator*{\argmin}{arg\,min}
\DeclareMathOperator*{\argmax}{arg\,max}
\renewcommand{\appendixtocname}{List of appendices}

\makeatletter
\let\oldappendix\appendices

\g@addto@macro\tableofcontents{%
  \let\tf@toc@orig\tf@toc
}
\renewcommand{\appendices}{%
  \renewcommand{\thesection}{\Roman{section}}
  \let\tf@toc\tf@app
  \addtocontents{app}{\protect\setcounter{tocdepth}{1}}
  \immediate\write\@auxout{%
    \string\let\string\tf@toc\string\tf@app
  }
  \oldappendix
}%

\g@addto@macro\endappendices{%
  \let\tf@toc\tf@toc@orig
  \immediate\write\@auxout{%
    \string\let\string\tf@toc\string\tf@toc@orig
  }%
}  

\newcommand{\listofappendices}{%
  \begingroup
  \renewcommand{\contentsname}{\appendixtocname}
  \let\@oldstarttoc\@starttoc
  \def\@starttoc##1{\@oldstarttoc{app}}
  \tableofcontents
  \endgroup
}

\makeatother

\usepackage[
    backend=biber,
    style=alphabetic,
  ]{biblatex}
\date{\vspace{-1cm}}

\title{The QAOA gets stuck starting from a good classical string}
\author{Madelyn Cain$^1$, Edward Farhi$^{2,3}$, Sam Gutmann, Daniel Ranard$^3$, Eugene Tang$^3$\\
${}^1$\it\small{Department of Physics, Harvard University, Cambridge, MA 02138}\\
${}^2$\it\small{Google Quantum AI, Venice, CA 90291}\\
${}^3$\it\small{Center for Theoretical Physics, Massachusetts Institute of Technology, Cambridge, MA 02139}}

\begin{document}

\maketitle

\vspace{0.5 cm}

\begin{abstract}
    The Quantum Approximate Optimization Algorithm (QAOA) is designed to maximize a cost function over bit strings. While the initial state is traditionally a  uniform superposition over all strings, it is natural to try expediting the QAOA: first use a classical algorithm to produce some good string, and then run the standard QAOA starting in the computational basis state associated with that string.  Here we report numerical experiments that show this method of initializing the QAOA fails dramatically, exhibiting little to no improvement of the cost function.  
    We provide multiple analytical arguments for this lack of improvement, each of which can be made rigorous under different regimes or assumptions, including at nearly linear depths. 
    We emphasize that our negative results only apply to our simple incarnation of the warm-start QAOA and may not apply to other approaches in the literature. 
    We hope that our theoretical analysis will inform future algorithm design. 
\end{abstract}
\vspace{-0.2 cm}
\textit{}

\thispagestyle{title}

\section{Introduction}\label{sec:intro}

The Quantum Approximate Optimization Algorithm (QAOA)~\cite{QAOA_2014} is designed to find good approximate solutions to the problem of maximizing a cost function defined on bit strings.  The standard QAOA starts in the uniform superposition of all classical bit string and applies cost function-dependent unitaries, with the goal of improving the cost function beyond its expected value in the initial state. It is natural to imagine running a classical algorithm to produce a good string then initializing the QAOA in the associated computational basis state
and looking for improvement.  Variations of this idea have been explored in existing literature under the name of \emph{warm-start QAOA}.

Here we seek to theoretically understand a counter-intuitive  phenomenon: in extensive numerical experiments involving small and large instances at low and high depths, the QAOA makes effectively \textit{zero} progress when starting from a single warm start string. That is, nearly every good classical string appears to be a local maximum that cannot be improved by \textit{any} QAOA circuit at any numerically accessible depth.
The consistency of our numerics suggests there is an underlying theoretical reason why the QAOA with a warm start gets stuck. In this work, we provide several analytic arguments to explain this phenomenon. Each argument has a different flavor and can be made rigorous in different regimes or under different assumptions.  Section \ref{sec:small_angles} presents an argument for small angles and low depths, Section \ref{sec:thermal} presents an argument for constant depths that rests on a thermality assumption, and Section \ref{sec:compression} presents an argument for sub-linear depths that rests on an assumption about how the number of strings decreases at higher cost.

We emphasize that our results apply to the case where the QAOA is initialized in a single classical string and the unitary operators comprising the QAOA do not depend explicitly on the initial string; they do not necessarily apply to  possibly more viable approaches to the warm-start QAOA in the literature~\cite{Tate_2020, Egger_2021}.

\section{Standard and single-string warm-start QAOA}\label{sec:definitions}

\subsection{A quick review of the standard QAOA}
The standard QAOA starts in the uniform superposition of all $n$-bit strings $\ket{s} = \frac{1}{2^{n/2}}\sum_{z}\ket{z}$, and applies a sequence of $2p$ unitaries depending on the $2p$ parameters $\boldsymbol{\gamma}$ and $\boldsymbol{\beta}$, producing
\begin{equation}
    \ket{\boldsymbol{\gamma}, \boldsymbol{\beta}} = e^{-i\beta_p B} e^{-i\gamma_p C}\dots e^{-i\beta_1 B}e^{-i\gamma_1 C} \ket{s},\label{eq:standard_QAOA_unitary}
\end{equation}
where $B = \sum_i X_i$ is the sum of the single-qubit $X$ operators, and $C$ is the diagonal cost function operator.  The goal is then to find the parameters $\boldsymbol{\gamma}, \boldsymbol{\beta}$ to make $\bra{\boldsymbol{\gamma}, \boldsymbol{\beta}} C\ket{\boldsymbol{\gamma}, \boldsymbol{\beta}}$ as large as possible. We focus on graph-based problems such as MaxCut,  where the cost function can be written as 
\begin{equation} \label{eq:cost_function}
    C(z) = \sum_{\langle i, j\rangle\in E}C_{\langle i, j\rangle}(z_i, z_j),
\end{equation}
where $\langle i, j \rangle$ ranges over all edges of a graph $G = (V, E)$. For example, the MaxCut cost function is 
\begin{align}
C_\mathrm{MC}=\frac{1}{2}\sum_{\langle i,j\rangle \in E}(1-Z_i Z_j),
\end{align}
which counts the number of cut edges. The cut fraction is defined as the number of cut edges divided by the total number of edges. 
It will often be convenient to consider the simpler MaxCut cost function operator,
\begin{align}
    C_Z = -\sum_{\langle i,j\rangle \in E}Z_i Z_j.
\end{align}

\subsection{The QAOA starting from a good classical string}
In this paper, we focus on the QAOA starting at a computational basis state $\ket{w}$, where $w$ is a good string, with operators in the same form as the previous section. Note that when we start in $\ket{w}$, the first unitary $e^{-i\gamma_1 C}$ in equation~\eqref{eq:standard_QAOA_unitary} only introduces a phase, so effectively the state becomes
\begin{equation} 
    \ket{\boldsymbol{\gamma}, \boldsymbol{\beta}, w} = U(\boldsymbol{\gamma},\boldsymbol{\beta})|w\rangle = e^{-i\beta_k B} e^{-i\gamma_{k-1} C}\dots e^{-i\gamma_1 C}e^{-i\beta_1 B} \ket{w},\label{eq:warm_start_QAOA_unitary}
\end{equation}
which depend on the $2k-1$ parameters $\beta_1,\dots, \beta_k$ and $\gamma_1,\dots, \gamma_{k-1}$. In keeping with the convention of the standard QAOA, where there are $2p$  parameters, we define the QAOA starting on a classical string with a fractional value of $p$ so that $2p = 2k-1$. 

Given $w$, the goal is to maximize $\bra{\boldsymbol{\gamma}, \boldsymbol{\beta}, w} C\ket{\boldsymbol{\gamma}, \boldsymbol{\beta}, w}$ by finding the optimal parameters
\begin{align}
    (\boldsymbol{\gamma}_w,\boldsymbol{\beta}_w) = \argmax_{(\boldsymbol{\gamma},\boldsymbol{\beta})}\,\langle w|U^\dagger(\boldsymbol{\gamma},\boldsymbol{\beta})CU(\boldsymbol{\gamma},\boldsymbol{\beta})|w\rangle.
\end{align}
We then define  the optimal QAOA unitary for the string $w$,
\begin{align} \label{eq:U_w}
    U_w = U(\boldsymbol{\gamma}_w,\boldsymbol{\beta}_w),
\end{align}
which depends on $w$ only through the optimal parameters $(\boldsymbol{\gamma}_w,\boldsymbol{\beta}_w)$.

\subsection{Cost function as a sum over neighborhoods} \label{subsec:nbhd_sums}

We now review how the expected cost function may be organized as a sum over subgraphs, as in~\cite{QAOA_2014}.  This is instructive for analyzing large graphs at constant $p$ in  Sections~\ref{subsec:large_num} and~\ref{sec:thermal}. We consider the warm-start QAOA with initial classical string $w$ and expected cost $\bra{w} U^\dagger C U \ket{w}$, where $U = U(\boldsymbol{\gamma},\boldsymbol{\beta})$ is the QAOA unitary.
In this regime, the operator $U^\dagger C_{\langle i,j\rangle}U$ only acts non-trivially on the ``edge neighborhood'' of  $\langle i,j\rangle$: vertices within graph distance $r$ of $\langle i,j\rangle$, where the radius $r=p-\frac{1}{2}$ is the number of applications of $e^{-i \gamma C}$ in the QAOA circuit. Because $U^\dagger C_{\langle i,j\rangle}U$ only acts non-trivially on the edge neighborhood of $\langle i, j\rangle$, we can write
\begin{align}
\langle w|U^\dagger C U|w\rangle &=\sum_{\langle i,j\rangle \in E}\langle w|U^\dagger C_{\langle i,j\rangle}U|w \rangle. \\
&=\sum_{\langle i,j\rangle \in E}\langle w_{\langle i,j\rangle}|U^\dagger C_{\langle i,j\rangle}U|w_{\langle i,j\rangle}\rangle. 
\end{align}
where $w_{\langle i,j\rangle}$ refer to the subset of bits in the corresponding edge neighborhood. 
For constant $p$ on a bounded-degree graph, this sum can be handled with low computational cost: each term can be computed using only the Hilbert space of the qubits in the edge neighborhood and the values $w_{\langle i, j\rangle}$. 

\section{Numerical experiments}\label{sec:num}

We first present numerical experiments for MaxCut that consistently show that running the standard QAOA starting from a good string does not lead to improvement, both for small graphs at high depth and also for larger graphs where we can only explore at  lower depth. The results we present are typical; we did numerical experiments for various graphs of different sizes and at many QAOA depths, and never saw appreciable improvements. We replicate the small system size results in Appendix~\ref{sec:MIS} for the Maximum Independent Set problem at $16$ vertices. In Appendix~\ref{sec:SK} we also look at the Sherrington-Kirkpatrick~(SK) model, which is fully connected and has random couplings, and obtain similar results at $14$ vertices, albeit with rare exceptions discussed in Appendix~\ref{sec:magic_angle}.

\subsection{Small system sizes and high depth\label{subsec:small_num}}

In Table~\ref{tab:warm_start_maxcut_12} we present data for  MaxCut on a 12 vertex, 3-regular graph at depths up to $p=9/2$, where each edge neighborhood wraps around the graph many times. For this instance, the largest cut value is 16 and the mean cut value is 9. 
For each warm start we optimize the QAOA using a random initial guess for parameters, and repeat the maximization for 40 random initial guesses.  We first optimize all strings with a cut value of 11, and find that most strings do not improve at all, and the rest have small improvement. 
We then look at all classical strings with a better cut value of 13.  Here even fewer strings show improvement, and again the improvements are small.  
We find that the QAOA never improves when initialized in  near-optimal strings with a cut value of 15.  We expect that the failure to find improvement at very low bit number will continue at high bit number, which we verify next. In Table~\ref{tab:standard_maxcut_12} we show how the usual QAOA which starts in the uniform superpositon makes steady progress to the optimal value as the depth increases.

\begin{table}[t!]
\begin{center}
    \begin{tabular}{r|cccc}
    \multicolumn{5}{c}{\textsc{Warm start at $C_{\text{MC}}(w) = 11$ with 516 classical strings $w$}}\\
    \hline
        $p$  & $3/2$ & $5/2$ & $7/2$ & $9/2$ \\
        Number of strings improved & 56 & 72 & 112 & 164 \\
        Mean cost of improved strings & 11.03 & 11.05 & 11.09 & 11.10\\
        Largest cost of improved strings & 11.04 & 11.18 & 11.29 & 11.42 \\
    \multicolumn{5}{c}{\rule{0pt}{2ex}} \\
    \multicolumn{5}{c}{\textsc{Warm start at $C_{\text{MC}}(w)= 13$ with 126 classical strings $w$}}\\
    \hline
        $p$  & $3/2$ & $5/2$ & $7/2$ & $9/2$ \\
        Number of strings improved & 2 & 2& 2 & 6 \\
        Mean cost of improved strings & 13.05 & 13.12 & 13.14 & 13.06 \\
        Largest cost of improved strings & 13.05 & 13.12 & 13.14 & 13.15 \\ 
    \multicolumn{5}{c}{\rule{0pt}{2ex}} \\
    \multicolumn{5}{c}{\textsc{Warm start at $C_{\text{MC}}(w)= 15$ with 10 classical strings $w$}}\\
    \hline
        $p$  & $3/2$ & $5/2$ & $7/2$ & $9/2$ \\
        Number of strings improved & 0 & 0 & 0 & 0 \\

    \end{tabular}
    \caption{\textit{MaxCut. QAOA improvement from good initial classical strings on a 12 vertex, 3-regular graph with a maximum cut size of $16$. The average cut size is $9$ and the number of parameters is $2p$. When warm starting at cut size $11$, the number of strings which improve  increases steadily with depth, although all improvements are small. All improvements come in pairs because of the bit flip symmetry of MaxCut.  At cut size $15$ we see zero improvements.}\label{tab:warm_start_maxcut_12}}
    \begin{tabular}{r|cccccc}
    \multicolumn{7}{c}{\rule{0pt}{2ex}} \\
    \multicolumn{7}{c}{\textsc{Standard QAOA}}\\
    \hline
        $p$  & $1$ & $2$ & $3$ & $4$ & $5$ & $6$ \\
        Expected cut size & 12.15 & 13.43 & 14.28 & 14.86 & 15.20 & 15.41 \\
        
    \end{tabular}
    \caption{\textit{MaxCut. Performance of the standard QAOA for the same instance as above, starting in the uniform superposition where the expected value of $C_{\text{MC}}$ is $9$. The number of parameters is $2p$. Note the steady improvement with depth.}\label{tab:standard_maxcut_12}}
    
\end{center}
\end{table}

\subsection{Large system sizes and low depths\label{subsec:large_num}}
We now look at MaxCut on large $3$-regular graphs. We toss a representative instance with \mbox{$n=300$} vertices and a largest cut size of $415$. We prepare good classical strings using two classical algorithms: simulated annealing 
and Goemans-Williamson~\cite{GW_algorithm}. 
Our simulated annealing algorithm samples strings from a thermal distribution using a Metropolis-Hastings cluster update rule $1000n$ times. For temperatures of $1.25$ and $1.75$,
we generate 1000 good classical strings with average cut size $310.6$ and $287.8$, respectively. We also produce 1000 classical strings with an average cut value of $398.3$ using the Goemans-Williamson algorithm. We then simulate the QAOA at $p=3/2$ and $p=5/2$ using the edge neighborhoods trick in  Section \ref{subsec:nbhd_sums}, and see zero improvement on any of the initial strings generated by either algorithm. Our results suggest that the fraction of strings which improve goes to zero as the system size grows.

\section{Failure on a Decoupled
 Problem \label{sec:decoupled_problem}}

A simple intuition for why starting from a good string might fail is as follows: consider the cost function $C_{\langle i, j\rangle}$ on edge $\langle i, j\rangle$ for MaxCut.  For a given string $w$, the clause is either satisfied or not, so we can suggestively write the expectation of $C$ after applying the QAOA unitary $U$ as
\begin{align}
    \bra{w} U^\dagger C_{\langle i, j\rangle}U\ket{w} = 1-\varepsilon_{\langle i, j\rangle }
\end{align}
 in the satisfied case, and
\begin{align}
    \bra{w} U^\dagger C_{\langle i, j\rangle}U\ket{w} = \delta_{\langle i, j\rangle }
\end{align}
in the unsatisfied case (note that if $U=1$ then $\varepsilon_{\langle i, j\rangle }$ and $\delta_{\langle i, j\rangle }$ are zero). Now imagine that the average of $\varepsilon_{\langle i, j\rangle }$ over the satisfied edges is equal to the average of $\delta_{\langle i, j\rangle }$ over the unsatisfied edges.  In a good string, there are more satisfied edges than unsatisfied, so we have
\begin{align}
    \bra{w} U^\dagger C U\ket{w} \le \bra{w}C\ket{w},
\end{align}
where $C$ is the total cost over all edges.
Therefore, no improvement would be possible. 

We now give a simple example where exactly this happens. Consider a graph of $2m$ vertices consisting of $m$ pairs of vertices, each pair connected by an edge. Focus on the first copy which has edge $\langle 1, 2\rangle$.  At any depth, the QAOA consists of applications of the unitaries that are exponentials of phases times $X_1+X_2$ or $Z_1Z_2$. Note that these operators commute with the swap operator on 1 and 2 as well as the operator $X_1 X_2$. The eigenstates and eigenvalues of these two operators are given in Table~\ref{tab:disconnected}.
\begin{table}[h!]
\begin{center}
    \begin{tabular}{c|cc}
    \hline
        Eigenstate  & Eigenvalue of swap & Eigenvalue of $X_1X_2$\\
        $\ket{a} = \frac{1}{\sqrt{2}}(\ket{+1,+1} - \ket{-1,-1})$& $+1$ & $-1$ \\
        $\ket{b} = \frac{1}{\sqrt{2}}(\ket{+1,-1} - \ket{-1,+1})$ & $-1$ & $-1$ \\
        $\ket{c} = \frac{1}{\sqrt{2}}(\ket{+1,+1} + \ket{-1,-1})$& $+1$ & $+1$  \\
        $\ket{d} = \frac{1}{\sqrt{2}}(\ket{+1,-1} + \ket{-1,+1})$& $+1$ &  $+1$ 
    \end{tabular}
    \caption{\textit{ Eigenstates and eigenvalues of the swap and $X_1X_2$ operators.}\label{tab:disconnected}}
\end{center}
\end{table}

Note that $\ket{a}, \ket{b}$ have unique quantum numbers so under any QAOA unitary they can each only pick up a phase, but $\ket{c}, \ket{d}$ can mix as (up to an overall phase)
 \begin{align}
     U\ket{c} &= \cos\theta \ket{c} +e^{i\phi} \sin\theta \ket{d},\\
     U\ket{d} &= e^{-i\phi}\sin\theta \ket{c} - \cos\theta \ket{d}.
 \end{align}
Then the QAOA initialized in the unsatisfied string $\ket{+1,+1}=  \frac{1}{\sqrt{2}}(\ket{c} + \ket{a}) $ evolves as
\begin{align}
    U\ket{+1,+1} = \frac{1}{\sqrt{2}}( \cos\theta \ket{c} + e^{i\phi} \sin\theta\ket{d} + e^{i\eta_1}\ket{a} )
\end{align}
where $\eta_1$ is another phase.
Therefore, the expected value of $C$ goes from 0 to $ \frac{1}{2}\sin^2\theta$.  Similarly, initializing the QAOA in the satisfied string $\ket{+1,-1}$, we find that the expected value of $C$ goes from 1 to $1-\frac{1}{2}\sin^2\theta$.
So for the $m$-edge problem starting in any string $w$, every unsatisfied clause is improved by $\frac{1}{2}\sin^2\theta $ and every satisfied edge degrades by the same amount.  Therefore, no matter  the depth, the QAOA starting in a good string cannot improve upon its initial cost.
 
We contrast this with how the usual QAOA performs.  Starting in the uniform superposition at $p=1$ with $\gamma=\pi/2$ and $\beta=\pi/8$, we go to the perfect cut with probability $1$. Note that for the standard QAOA, there is an asymptotic performance guarantee in that the maximum expectation value is guaranteed to reach the optimum as $p\rightarrow \infty$~\cite{QAOA_2014}. This example is notable in showing that no such guarantee is present for the case of warm-start QAOA.  

\section{Small angles and low depth}\label{sec:small_angles}


When improvements are seen in the numerical experiments, they are always small. One explanation for this could be that the optimal $\beta$'s, which drive the mixing, are also small. Therefore, in this section we examine when improvements are possible for small $\beta$'s. First we show that for MaxCut at $p=1/2$, no improvement is possible for any $\beta_1$. Next we give two necessary and sufficient conditions for improvement at $p=3/2$ when $\beta_1, \beta_2$ are small.  We find that in our MaxCut numerical experiments at $p=3/2$, improvement was seen if and only if one of our conditions were satisfied. Consistent with our intuition, improvement is possible only if it is possible for small $\beta$'s. 

\subsection{Quantum walk starting on a classical string ($p=1/2$)}
The expected cost for $p=1/2$ after evolving under $B$ is given by
\begin{align}
    \bra{w}e^{i\beta_1 B}C_Ze^{-i\beta_1 B}\ket{w} &= -\sum_{\langle i, j \rangle}\bra{w} e^{i\beta_1 B} Z_i Z_j  e^{-i\beta_1 B}\ket{w} \nonumber \\
    &=-\sum_{\langle i, j \rangle}\bra{w} (\cos(2\beta_1)Z_i +\sin(2\beta_1) Y_i )( \cos(2\beta_1)Z_j +\sin(2\beta_1) Y_j )  \ket{w}\nonumber \\
    &=\cos^2(2\beta_1) C_Z(w)\label{eq:maxcut_one_parameter}.
\end{align}
Therefore, if $w$ is a good starting string, meaning that $C_Z (w) > 0$,
then $\beta_1=0$ is optimal and the optimized expected cost remains $C_Z(w)$. 

\subsection{Warm starting at $p=3/2$ for MaxCut on $3$-regular graphs}

We now look at the $p=3/2$ QAOA warm starting in $|w\rangle$ with both $\beta$'s  small but $\gamma$ arbitrary. We focus on MaxCut for 3-regular graphs. We state a necessary and sufficient condition for improvement here, and give the detailed derivation  in Appendix~\ref{sec:small_angles_proofs}.

First, let us define the quantity
\begin{align}
    \delta_i = (C_i - C(w))/2\label{eq:definition_of_delta_i},
\end{align}
where $C_i = \bra{w} X_i C X_i \ket{w}$ measures the cost of $w$ after flipping bit $i$. For MaxCut on 3-regular graphs, $|\delta_i| = $ 1 or 3.  At $p=3/2$, we find that improvement of the cost is possible for small $\beta$'s if and only if
\begin{align}
 \sum_{|\delta_i|=1}\delta_i >0\qquad\text{or}\qquad \sum_i\delta_i^3 >0. \label{eq:sam_iff}
\end{align}

If $w$ is a local maximum with respect to bit flips, then all of the $\delta_i$ are negative and neither condition can  be met. We expect that for most good strings, neither of these two conditions can be met. We make use of the identity
\begin{align}
    \sum_{i}\delta_i = -2C_Z(w),\label{eq:delta_sum2}
\end{align}
which is straightforward to establish. Because $C_Z(w) > 0$ for a good starting string, the $\delta_i$ are usually negative, meaning that both conditions~\eqref{eq:sam_iff} are unlikely.  As the goodness of the starting string increases there will be fewer positive $\delta_i$, so the conditions are even less likely to be met. This is consistent with the numerics on small MaxCut instances at $p=3/2$ in Section~\ref{subsec:small_num} where there are fewer improvements on the better strings. Further, we expect that as $n$ increases, fluctuations that allow the conditions to be met will become even more rare. This is consistent with the findings in Section~\ref{subsec:large_num} on large MaxCut instances, where the condition is never met and improvement is never seen for thousands of initial strings. In all of our numerics in Section~\ref{sec:num}, improvement was seen at $p=3/2$ if and only if equation~\eqref{eq:sam_iff} was satisfied.

\section{Thermal argument for getting stuck\label{sec:thermal}}

In this section we give two upper bounds for the progress made by the single-string warm-start QAOA at constant depth on large bounded-degree graphs. The first bound shows that strings chosen uniformly at random are likely to get stuck.  The second bound instead considers strings chosen from a finite-temperature thermal ensemble, and the argument uses the
principle of minimum energy from thermodynamics in equation~\eqref{eq:min_energy}. 


\subsection{Starting from a uniformly random string \label{sec:thermal_1}}

We first consider the case of starting the QAOA with a classical bit string chosen uniformly at random. Our results apply to generic cost functions for general bounded-degree graphs,
\begin{equation} 
    C(z) = \sum_{\langle i, j\rangle\in E}C_{\langle i, j\rangle}(z_i, z_j),\label{eq:generic_cost_thermal}
\end{equation}
with $m$ edges at any constant depth. We define the average value per edge of the cost function, 
\begin{align}
    \overline{c} = \frac{1}{2^n}\sum_{z\in\{-1,1\}^n}\frac{1}{m} C(z).\label{eq:avg_cost_per_edge}
\end{align}
Using the fact that $C$ is a sum of terms, most of which are uncorrelated because of the bounded degree, one can show that for any $\epsilon>0$, and $w$ chosen uniformly at random,
 \begin{align} 
    \Pr_w\left(\left|\frac{1}{m} C(w) - \bar{c} \right| \ge \frac{1}{m^{1/2 - \varepsilon}}\right) \longrightarrow 0,\quad\text{as }m\rightarrow\infty. \label{eq:controverial_equation}
\end{align}
This implies that in the limit of large graph size $m$, almost all strings have cost function values extremely close to the average value $\overline{c}$.

In Appendix~\ref{subsec:average_string_failure} we prove the following theorem, which states that typical strings cannot be improved by constant depth QAOA.
\begin{theorem}\label{thm:no_improvement_typical}
For any $\epsilon>0$, and $w$ chosen uniformly at random,
\begin{align}
    \Pr_w\left(\frac{1}{m}\langle w|U_w^\dagger CU_w |w\rangle - \frac{1}{m}C(w) \ge \frac{1}{m^{1/2-\varepsilon}}\right) \rightarrow 0,\qquad\text{as }m\rightarrow\infty,\label{eq:chi_zero}
\end{align}
where $U_w$ is the optimal constant-depth QAOA unitary for each $w$.
\end{theorem}
Therefore, we have that for nearly all initial strings, the constant-depth QAOA can only improve the cost function (per edge) by an amount that vanishes in the limit of large graphs. While this result does not apply to warm starts, whose starting cost value per edge $C(w)/m$ is substantially larger than $\bar{c}$, it is very surprising that even starting from a mediocre string, the QAOA fails to make progress. 

Equation~\eqref{eq:controverial_equation} implies that warm starts form a vanishingly small fraction of all strings. Next, we give results on these rare warm-start strings. 

\subsection{Starting from a typical good string \label{sec:thermal_2}} 

We work in the setting of $d$-regular graphs with a cost function of the form of equation~\eqref{eq:generic_cost_thermal}.
For typical large random  $d$-regular graphs, almost all of the local neighborhoods are trees. Fixing a neighborhood radius $r$ (corresponding to the QAOA with $p=r+\frac{1}{2}$, see Section~\ref{subsec:nbhd_sums}), we define the fraction $\delta$ of edges whose neighborhoods are \textit{not} trees as
\begin{align}  
\delta = 1-\frac{|E_T|}{m} ,
\end{align}
where $E_T$ denotes the set of edges whose neighborhood of radius $r$ is a tree. Our argument will be useful in the case that $\delta$ is small. For random $d$-regular graphs with $n \to \infty$ with $d$ and $r$ fixed,  $\delta = O\left(1/n\right)$~(\cite{janson_random_regular}, Theorem 9.5) is vanishingly small.

We introduce two ensembles associated with the warm-start string $w$. For any classical string \mbox{$w\in \{-1,1\}^n$} and a mostly locally tree-like $G$ we define its \emph{local ensemble} by
\begin{align}
    \rhowtree = \frac{1}{|E_T|}\sum_{\langle i,j\rangle \in E_T}|w_{\langle i,j\rangle}\rangle\langle w_{\langle i,j\rangle}|,
\end{align}
where $w_{\langle i,j\rangle}$ denotes the restriction of $w$ onto the local tree neighborhood $T_{\langle i,j\rangle}$, and  $T_{\langle i,j\rangle}$ denotes the tree neighborhood centered around edge $\langle i,j\rangle \in E_T$. Each tree neighborhood has a fixed number of vertices, so $w_{\langle i,j \rangle}$ is a string of that length, and $\rhowtree$ is a density matrix over the corresponding number of qubits.

Let
\begin{align} \label{eq:rhobeta_defn}
\rho_{\beta} = \frac{1}{\mathrm{Tr}\big(e^{-\beta(- C)}\big)}e^{-\beta(- C)}
\end{align}
be the global thermal density matrix with respect to the Hamiltonian $H=-C$, at inverse temperature $\beta$ chosen so that  $
\mathrm{Tr}(C\rho_\beta) = C(w) $. Here we introduce a minus sign because in thermodynamics we associate low temperatures with low energies, whereas for us good strings have high $C$.

We will consider the reduced density matrix of $\rho_\beta$ onto the tree neighborhood $T_{\langle i,j\rangle}$, denoted $\rhobetaij$.
Because all the tree neighborhoods $T_{\langle i,j\rangle}$ are isomorphic subgraphs, we can identify their associated Hilbert spaces. We consider each reduced operator $\rhobetaij$ to act on the same abstract Hilbert space (the same Hilbert space implicit in the definition of $\rhowtree$). With this identification, we can define the average over neighborhoods,
\begin{align} \label{eq:rhobetatree}
    \rhobetatree = \frac{1}{|E_T|}\sum_{\langle i,j\rangle \in E_T} \rhobetaij.
\end{align}
Again, both $\rhowtree$ and $\rhobetatree$ act on the Hilbert space of a single tree neighborhood: although their definitions involve the sum of terms associated to different neighborhoods, these neighborhoods are all isomorphic trees, and we implicitly identify the neighborhoods when summing the terms.

We now give our main result, an upper bound on how much the cost function of a warm-start string $w$ can improve under the QAOA at fixed depth. The upper bound depends on a quantity $\varepsilon_w$, which we call the ``thermality coefficient'', that measures the distance between the ensembles $\rhowtree$ and $\rhobetatree$:
\begin{align}
\label{eq:local_therm}
    \varepsilon_w =  \| \rhobetatree- \rhowtree\|_1,
\end{align}
where $\beta$ is the inverse temperature chosen so that 
\begin{align}
\mathrm{Tr}(C\rho_\beta) = C(w)
\end{align} 
and where 
$\|\rho-\sigma\|_1 \equiv \mathrm{Tr}\Big(\sqrt{(\rho-\sigma)^\dagger(\rho-\sigma)}\Big)$
denotes the trace norm. Both $\rhobetatree$ and $\rhowtree$ are diagonal in the computational basis, i.e., they are classical probability distributions over strings, in which case $\varepsilon_w$ is also equal to twice the total variation distance. When $\varepsilon_w$ is small, we refer to $w$ as being ``locally thermal''. Note that $\varepsilon_w$ depends on the QAOA depth $p$ which defines the neighborhood size. Our main result, stated next and proven in Appendix~\ref{subsection:thermal_proof}, shows that if $w$ is locally thermal, i.e., if $\varepsilon_w$ is small, then $w$ can only improve by a small amount. 
\begin{theorem} \label{thm:cbar0}
Consider the single-string warm-start QAOA initialized in a string $w$ with thermality coefficient $\varepsilon_w$, and let $\delta$ be the fraction of neighborhoods whose edges are not trees. Then
    \begin{align} 
     \frac{1}{m}\langle w|U^\dagger_w C U_w|w\rangle  \leq c(w) +2\varepsilon_w + 4\delta ,
\end{align}
where $c(w)= C(w)/m$ is the cut fraction of the original string $w$.
\end{theorem}

Because $\delta$ is small for random graphs, the heart of the result is the dependence on $\varepsilon_w$. If $w$ is locally thermal, then $w$ can only improve by a small amount. In Appendices~\ref{sec:why_small} and~\ref{sec:thermal_num}, we give theoretical and numerical evidence that $\varepsilon_w$ scales as $ 1/\sqrt{m}$ for typical good strings $w$.

\section{Compression argument for getting stuck\label{sec:compression}}
In this section we develop another argument to upper bound the progress made by the QAOA starting from a good string. This result is complementary to the previous results and uses a different assumption. We imagine that the initial good string is chosen uniformly at random from the set of all strings with cost $C_0$. We seek to improve the cost above some target $C_1$. Our argument will produce a useful bound whenever the number of strings $d_0$ at cost  $C_0$ is much larger than the number of strings $d_1$ with cost at least $C_1$.  In particular, if the ratio $d_1/d_0$ is suppressed exponentially in the size of the graph, then with high probability, the randomly chosen initial string cannot be improved beyond $C_1$ using depth $p = O(n^q)$ for $q<1$. 

We refer to this argument as a ``compression'' argument because of the following intuition. Consider a fixed QAOA unitary $U$.  If there are many more strings at cost  $C_0$ than at cost above $C_1$, then due to unitarity, $U$ cannot map all (or even most) of the strings from cost $C_0$ to above $C_1$. That is, the unitary cannot ``compress'' a larger space (of strings at $C_0$) to a smaller space (of strings above $C_1$). This is the core of the argument. Of course, there is more than just one QAOA unitary, so there is no need for a single unitary $U$ to map all the strings at $C_0$ to  above $C_1$.  Instead, one can use different QAOA unitaries $U_w$ that depend on the initial string $w$, hoping that for each $w$ at cost $C_0$ there is \textit{some} $U_w$ such that $U_w |w\rangle$ has cost at least $C_1$.  However, this is impossible when the depth $p$ is not too large, because the space of unitaries that can be generated by the QAOA at depth $p$ is small in the appropriate sense.


Our results prior to this section only characterize the improvement in expected cost, without explicitly bounding the probability of obtaining a good string after measurement.  However, for the case of states produced by constant-depth circuits, one can show the distribution of measured costs is concentrated about the expected cost~\cite{QAOA_2014}. In contrast, here we will consider circuits of super-constant depth, so we must take care with the distinction between the expected cost and the distribution of measured costs.
To this end, we introduce the following notion of \emph{improvable strings}.

\begin{definition}\label{def:improvable}
A string $w$ is $p$-improvable to cost $C_1$, with a selection probability at least $\Delta$, if there exists a depth-$p$ QAOA unitary $U_w$ such that a measurement of $U_w |w\rangle$ in the computational basis has a probability at least $\Delta$ of yielding a string with cost at least $C_1$.
\end{definition}


\begin{theorem}\label{thm:improvable_main}
Consider warm-start QAOA at depth $p$, starting with a string chosen uniformly at random among strings with cost $C_0$. Let $d_0$ be the number of strings at cost $C_0$, and let $d_1$ be the number of strings with cost at least $C_1$.  Then with probability at least $1-\epsilon$, the initial string will not be improvable to cost $C_1$ using selection probability at least $\epsilon$, where
\begin{align}
    \epsilon = \left(\frac{2d_1}{d_0}\right)^{\frac{1}{2p+2}}(16 \pi p n^2)^{\frac{p}{p+1}}.
\end{align}
\end{theorem}
This theorem is most useful when $d_1/d_0$ is exponentially suppressed in $n$, i.e.\, when there are exponentially more strings near the starting cost than above the target cost.  Then $\epsilon$ will be small for $p$ sub-linear in $n$, and we conclude that with high probability, the initial string will not be improvable. 

More concretely, suppose that there are exponentially (in $n$) more strings at cost $C_0$ than there are strings with cost above $C_1$. This is often the case when $C_1-C_0$ scales with $n$, and we provide further justification below. In that case we have 
\begin{align} \label{eq:d1d0_scaling}
    \frac{d_1}{d_0} = O(e^{-a n})
\end{align}
for some constant $a>0$. Consider a depth $p = O(n^q)$ for $q < 1$.  Then Theorem \ref{thm:improvable_main} holds with 
\begin{align}
    \epsilon = O(e^{-b n^{1-q}}),
\end{align}
for some constant $b>0$, decaying faster than any polynomial in $n$.  Thus the probability of improvement is very small at large $n$.

To further justify the scaling of \eqref{eq:d1d0_scaling}, we consider the MaxCut cost function of the decoupled graph in Section~\ref{sec:decoupled_problem}, which has $m=n/2$ pairs of vertices, each pair connected by an edge.  The number of strings as a function of cost (i.e., the density of states) is binomial, with variance $m$. For $C_1-C_0 \propto m$, one obtains \eqref{eq:d1d0_scaling}. More generally, one also expects \eqref{eq:d1d0_scaling} whenever the free energy associated to $C$ is extensive since $S = -\partial F/\partial T$ then implies an extensive entropy. The free energy of MaxCut on $d$-regular graphs exhibits this behavior, as shown (at least above the phase transition) in \cite{coja2022ising}. For other classes of problems with random local terms, also see Proposition 4 of \cite{dalzell2022mind}, where their upper bound on the number of solutions as a function of cost indicates similar exponential scaling.

\section{Conclusion}\label{sec:conclusion}
This work was motivated by  numerical experiments on the single-string warm-start QAOA showing zero improvement for almost all good initial strings from the starting cost value. We presented several theoretical results to explain these consistent observations. 
This behavior stands in stark contrast to the standard QAOA initialized in the uniform superposition over all bit strings, which appears to make consistent progress toward the optimum.



\textbf{Acknowledgments}

M.C. acknowledges support from DOE CSG award fellowship (DESC0020347). D.R. acknowledges support from NTT (Grant AGMT DTD 9/24/20). E.T. acknowledges funding received from DARPA 134371-5113608, and DOD grant award KK2014. M.C. and E.F. thank Boaz Barak, Beatrice Nash, and Leo Zhou for early  discussions of these ideas. D.R. and E.T. thank Aram Harrow for discussion.

\printbibliography

@misc{QAOA_2014,
  author = {Farhi, Edward and Goldstone, Jeffrey and Gutmann, Sam},
  title = {A Quantum Approximate Optimization Algorithm},
  year = {2014},
  archivePrefix = {arXiv},
  eprint = {1411.4028},
  primaryClass = {quant-ph}
}

@article{Egger_2021,
	doi = {10.22331/q-2021-06-17-479},
	year = 2021,
	month = {June},
	publisher = {Verein zur Forderung des Open Access Publizierens in den Quantenwissenschaften},
	volume = {5},
	pages = {479},
	author = {Daniel J. Egger and Jakub Mare{\v{c}}ek and Stefan Woerner},
	title = {Warm-starting quantum optimization},
	journal = {Quantum},
	archivePrefix = {arXiv},
	eprint = {2009.10095},
	primaryClass = {quant-ph}
}

@misc{Tate_2020,
  author = {Tate, Reuben and Farhadi, Majid and Herold, Creston and Mohler, Greg and Gupta, Swati},
  title = {Bridging Classical and Quantum with \mbox{SDP} initialized warm-starts for \mbox{QAOA}},
  year = {2020},
  archivePrefix = {arXiv},
  eprint = {2010.14021},
  primaryClass = {quant-ph},
}

@article{GW_algorithm,
author = {Goemans, Michel X. and Williamson, David P.},
title = {Improved Approximation Algorithms for Maximum Cut and Satisfiability Problems Using Semidefinite Programming},
year = {1995},
issue_date = {Nov. 1995},
publisher = {Association for Computing Machinery},
address = {New York, NY, USA},
volume = {42},
number = {6},
issn = {0004-5411},
doi = {10.1145/227683.227684},
journal = {J. ACM},
month = {November},
pages = {1115–1145},
numpages = {31}
}

@inbook{janson_random_regular,
publisher = {John Wiley \& Sons, Ltd},
author = {Janson, Svante and Luczak, Tomasz and Rucinski, Andrzej},
isbn = {9781118032718},
title = {Random Regular Graphs},
booktitle = {Random Graphs},
chapter = {9},
pages = {233-270},
doi = {https://doi.org/10.1002/9781118032718.ch9},
year = {2000}
}

@article{coja2022ising,
  title={The Ising antiferromagnet and max cut on random regular graphs},
  author={Coja-Oghlan, Amin and Loick, Philipp and Mezei, Bal{\'a}zs F and Sorkin, Gregory B},
  journal={SIAM Journal on Discrete Mathematics},
  volume={36},
  number={2},
  pages={1306--1342},
  year={2022},
  publisher={SIAM}
}

@article{dalzell2022mind,
  title={Mind the gap: Achieving a super-Grover quantum speedup by jumping to the end},
  author={Dalzell, Alexander M and Pancotti, Nicola and Campbell, Earl T and Brand{\~a}o, Fernando GSL},
  journal={arXiv preprint arXiv:2212.01513},
  year={2022}
}

@article{SK_2019,
  author = {Farhi, Edward and Goldstone, Jeffrey and Gutmann, Sam and Zhou, Leo},
  title = {The Quantum Approximate Optimization Algorithm and the\\ Sherrington-Kirkpatrick Model at Infinite Size},
  archivePrefix = {arXiv},
  eprint = {1910.08187},
  primaryClass = {quant-ph},
  journal = {Quantum},
  volume = {6},
  pages = {759},
  year = 2022,
  month = {July},
  publisher = {Verein zur Forderung des Open Access Publizierens in den Quantenwissenschaften},
  doi = {10.22331/q-2022-07-07-759}


}

\appendix
\begin{appendices}
\newpage
\setcounter{tocdepth}{1} 
\listofappendices

In these appendices, we provide additional details and prove statements in the main text. In Section~\ref{sec:other_types}, we discuss other types of initial states for the QAOA, including starting from a bad string and starting from a uniform superposition of strings with the same cost. Next, we give evidence that the single-string warm-start QAOA fails on two new problems: the Maximum Independent Set problem on $d$-regular graphs (Section~\ref{sec:MIS}) and the Sherrington-Kirkpatrick model (Section~\ref{sec:SK}). In Section~\ref{sec:small_angles_proofs} we give conditions under which the single-string warm-start QAOA can improve  at $p=3/2$ for small mixing angles. In Section~\ref{sec:thermal_proofs} we prove that strings which are locally thermal cannot be improved. We also give analytic and numerical evidence that warm strings are locally thermal.  Finally, in Section~\ref{sec:compression_proofs} we prove our compression argument, which shows that the single-string warm-start QAOA cannot improve for sub-linear depths when the number of strings at a higher cost is very small compared to the number of strings at the initial cost.

\section{Using other types of starts\label{sec:other_types}}

\subsection{Starting from a bad string}

A bad string is one whose cost function per edge $C(w)/m$ is less than the average cost per edge $\bar{c}$, defined as 
\begin{align}
    \overline{c} = \frac{1}{2^n}\sum_{z\in\{-1,1\}^n}\frac{1}{m} C(z).
\end{align}
In contrast to good strings, bad strings can all be improved by the lowest depth QAOA by applying the one-parameter $p=1/2$ QAOA with $\beta=\pi/4$.  This rotates the state into a uniform superposition over all computational basis states, each with an amplitude of $2^{-\frac{n}{2}}$ times a phase that depends on $w$. Therefore, $C(w)$ improves to the mean of the cost function over all strings. 
The thoughtful reader can check that this does not serve as a counterexample to our claims.

\subsection{Starting from superpositions of strings with the same cost}\label{sec:uniform}

What about starting in a uniform superposition of basis states with the same good value of the cost function? It may be difficult to use a classical algorithm to produce such a superposition, but here we simply explore what happens if you have such a state in hand.

We randomly pick a $16$ bit 3-regular graph whose largest cut value is $22$. The mean over all strings of the MaxCut cost function is $12$.  We begin with the uniform superposition of all strings whose cost is exactly $12$. The overlap (amplitude squared) of this normalized state with the uniform superposition over all strings $|s\rangle$ is $0.162$. 
At $p=3/2$, optimal parameters take the cost function to $12.064$. At $p=5/2$ and $7/2$, the improvement is to $12.173$ and $12.246$, respectively. This is not much improvement and should be contrasted with the usual QAOA starting state $|s\rangle$, where the expected value of the cost is $12$ and where the QAOA makes steady good progress as the depth increases.  For warm-start superpositions at higher values of the cost, such as $13$ up to $18$, we see no improvements at all for $p=3/2,\ 5/2$, and $7/2$.


\section{Maximum Independent Set}\label{sec:MIS}
In this section, we show using another example that the QAOA gets stuck starting on a good string.  We look at the Maximum Independent Set problem: given a graph $G=(V,E)$, the goal is to find a maximum size subset of vertices which have no edge connecting any pair of them. Formally, we want to find a string of bits $b$ such that the Hamming weight $W(b) = \sum_i b_i$ is maximized, subject to the constraint that $\sum_{\langle i,j\rangle \in E}b_ib_j = 0$. Note that in this section we are using bits $b_i$ that are $\{0,1\}$-valued, with corresponding quantum operators $\hat{b}_i = (I-Z_i)/2$. To make this  problem amenable to the usual QAOA, where any bit string is a valid input, we use the relaxed cost function
\begin{align}
    C(b) = \sum_{i\in V}b_i - \sum_{\langle i,j\rangle \in E}b_ib_j \equiv W(b)-K(b), \label{eq:MIS_C}
\end{align}
which we aim to maximize. Given a bit string $b$ with cost $C(b)$, we can prune it, fixing violations $b_ib_j = 1$ one at a time by setting either $b_i$ or $b_j$ equal to 0. Fixing one violation decreases $W$ by $1$ but also decreases $K$ by at least $1$, so $C$ does not decrease. Continuing to prune until there are no violations, we obtain an independent set of size at least $C(b)$. This also shows that the maximum of $C$ is the size of the largest independent set.
Note that for a $d$-regular graph we can write the cost function as
\begin{align}
    C(b) = W(b)-K(b) = \sum_{\langle j,k\rangle \in E}\left[\frac{1}{d}(b_j+b_k) - b_j b_k\right],
\end{align}
which shows that this cost function is a sum of terms which act the same way on each edge, satisfying the uniformity assumption required for our statistical argument in Section~\ref{sec:thermal_2}.

We now provide numerical and analytic evidence that the single-string warm-start QAOA fails for Maximum Independent Set. We first see numerically at high $p$ that the warm-start QAOA fails at low bit number. Then, we show analytically that improvement is impossible at $p=1/2$ if the QAOA is initialized in an independent set obtained by running a standard greedy algorithm.

\subsection{Numerical experiments}
We randomly sample a $16$ vertex $3$-regular graph whose largest independent set size is $7$, with degeneracy $3$. The average over all strings of the cost $C(b)$ is $2$, which is therefore also the expected value of $C$ in the quantum state $|s \rangle$, where the usual QAOA starts.  In Table~\ref{tab:standard_mis_16} we show how the usual QAOA performs when starting with the state $|s \rangle$.  Note that as $p$ gets bigger, performance improves, steadily approaching the largest independent set possible which has size $7$.

\begin{table}[t!]
    \begin{center}
    \begin{tabular}{r|cccc}
    \multicolumn{5}{c}{\textsc{Warm start at $C_{\text{MIS}}(w) = 4$ with 371 independent sets $w$}}\\
    \hline
        $p$  & $3/2$ & $5/2$ & $7/2$ & $9/2$  \\
        Number of strings improved & 20 & 101 & 283 & 368  \\
        Mean cost of improved strings & $4.06$ & $4.06$ & 4.11 & 4.23 \\
        Largest cost of improved strings & $4.10$ &  $4.27$ & 4.47 & 4.64 \\
        \multicolumn{5}{c}{\rule{0pt}{2ex}} \\
    \multicolumn{5}{c}{\textsc{Warm start at $C_{\text{MIS}}(w)= 5$ with 230 independents sets  $w$}}\\
    \hline
        $p$  & $3/2$ & $5/2$ & $7/2$ & $9/2$   \\
        Number of strings improved & 0 & 0 & 1 & 1 \\
        Mean cost of improved strings & $-$ & $-$ & 5.002 & 5.002 \\
        Largest cost of improved strings & $-$ & $-$ & 5.002 & 5.002 \\ 
    \end{tabular}\\
    \caption{\textit{Maximum Independent Set. QAOA improvement from good initial classical strings (corresponding to valid independent sets) on a 16 vertex, 3-regular graph with a largest independent set size of 7. The number of parameters is $2p$. }   \label{tab:warm_start_mis_16}}
    \begin{tabular}{r|cccccc}
    \multicolumn{5}{c}{\rule{0pt}{2ex}} \\
    \multicolumn{5}{c}{\textsc{Standard QAOA}}\\
    \hline
        $p$  & $1$ & $2$ & $3$ & $4$ \\
        Expected cost & 4.40 & 5.14 & 5.61 & 5.86\\
    \end{tabular}
    \end{center}
    \caption{\textit{Maximum Independent Set. Performance of the standard QAOA for the same instance as Table~\ref{tab:warm_start_mis_16}, starting in the uniform superposition where the expected value of $C_{\text{MIS}}$ is $2$. The number of parameters is $2p$. Note that there is no indication that the standard QAOA gets stuck.}   \label{tab:standard_mis_16}}
\end{table}

We now look at what happens when starting with a good string.  We only look at starting strings which are valid independent sets, which means that $K(w)=0$. Table~\ref{tab:warm_start_mis_16} shows data for starting the QAOA with strings of cost $4$ and $5$.  Note that there is either no improvement or very small improvement, similar to our numerical experiments for MaxCut in Section~\ref{subsec:small_num}.

\subsection{No improvement at $p=1/2$}

We now analyze the performance of the $p=1/2$ QAOA on a $d$-regular graph starting from a string $w$. Recall that in this section the bits $w_i$ are $\{0,1\}$-valued. Using \eqref{eq:MIS_C}, for all $\beta_1$, we have 
\begin{align}
    \bra{w}e^{i\beta_1 B}Ce^{-i\beta_1 B}\ket{w} 
    &= \sum_{i\in V}\bra{w} e^{i\beta_1 X_i} \hat{b}_i  e^{-i\beta_1 X_i}\ket{w}- \sum_{\langle i, j \rangle\in E}\bra{w} e^{i\beta_1 X_i} \hat{b}_i e^{-i\beta_1 X_i}e^{i\beta_1 X_j} \hat{b}_j  e^{-i\beta_1 X_j}\ket{w}\\
    &= \sum_{i\in V}\bra{w} \left(\sin^2\beta_1 +[1-2\sin^2\beta_1]\hat{b}_i\right)\ket{w}\nonumber\\
    &\quad -\sum_{\langle ij \rangle\in E}\bra{w} \left(\sin^2\beta_1 +[1-2\sin^2\beta_1]\hat{b}_i\right)\left(\sin^2\beta_1 +[1-2\sin^2\beta_1]\hat{b}_j\right) \ket{w} \\
    &=
    W(w)-K(w) +\frac{d}{2}\left[4W(w)-n  \right]\sin^4\beta_1+[n-W(w)(d+2)]\sin^2\beta_1 +K(w) \sin^2(2\beta_1) 
    \label{eq:mis_one_parameter}.
\end{align}
Let us now assume that we start with a string $w$ which corresponds to a valid independent set, meaning that $K(w)=0$. The corresponding cost function value is then just $W(w)$, the Hamming weight of the initial string $w$. We can optimize over $\beta_1$ noting that we have a quadratic function of  $\sin^2\beta_1$. The optimum occurs when
\begin{align}
\sin^2\beta_1 &=\frac{1}{d}\frac{n-W(w)(d+2)}{n-4W(w)}.\label{eq:optimal_beta_MIS_1/2}
\end{align}
Substituting this value back into equation~\eqref{eq:mis_one_parameter} to get the optimal value of the cost function, we find
\begin{align}
    \bra{w}e^{i\beta_1 B}Ce^{-i\beta_1 B}\ket{w}
    &= W(w)+\frac{1}{2d} \frac{(n-W(w)(d+2))^2}{n-4W(w)}.
   \label{eq:mis_1/2_beta_new_optimum}
\end{align}
For improvement we need the second term in~\eqref{eq:mis_1/2_beta_new_optimum} to be positive. This requires $W<n/4$.  Since $\sin^2\beta_1$ is also positive, we require the stronger condition $W<n/(d+2)$.  Therefore, improvement can only be made if $C(w) = W(w) < n/(d+2)$. 

Now, suppose that our starting string is an independent set obtained by running a standard greedy algorithm which achieves an independent set of size at least $n/(d+1)$. Then $K = 0$ and $W\ge n/(d+1)$, so the condition is not met, and no improvement is possible for $p=1/2$. Note that starting with the empty set, that is with $W=0$, the $p=1/2$ QAOA will get to an independent set of size $n/(2d)$.

\section{The Sherrington-Kirkpatrick model}\label{sec:SK}
The Sherrington-Kirkpatrick (SK) model is defined on the complete graph, and on each edge there is a coupling with a random sign: 
\begin{align}
    C_{\text{SK}} = \sum_{i, j}J_{i,j}Z_i Z_j,
\end{align}
where each $J_{i,j}$ is $+1$ or $-1$ with probability $1/2$. Parisi discovered, for any typical set of $J$'s, that the maximum energy divided by $n^{\frac{3}{2}}$, in the limit as $n$ goes to infinity, is a computable constant approximately equal to $0.763166$. It is known that the QAOA starting in the uniform superposition $|s\rangle$ makes good progress toward the optimum as the depth increases~\cite{SK_2019}. 

Here we are on a complete graph, and each edge sees a different neighborhood because of the varying $J$'s. 
Because our thermal argument and small angles arguments now do not apply, it is interesting to confirm numerically that the single-string warm-start QAOA fails on Sherrington-Kirkpatrick model. Therefore, we look numerically at a random instance generated at $14$ bits.  The lowest energy is $-33$ and the highest is $+35$. The performance of the usual QAOA starting on the state $|s\rangle$ is shown in Table~\ref{tab:standard_sk_14}. Turning to warm starts in Table~\ref{tab:warm_start_sk_14}, we look at $50$ randomly chosen starting strings with energy $11$. We see zero or little improvement out to depth $p=7/2$ for all strings besides four notable exceptions which we will discuss below. At an energy of $17$, there are no unusual improvements and the number of strings which improve is fewer than at $11$. Aside from the exceptions which we explain next, warm starts in the SK model also get stuck as they do in our other examples.  

\begin{table}[t!]
    \begin{center}
    \begin{tabular}{r|ccc}
    \multicolumn{4}{c}{\textsc{Warm start at $C_{\text{SK}}(w) = 11$ }}\\
    \multicolumn{4}{c}{746  classical strings $w$, 50 samples, 4 exceptions removed}\\
    \hline
        $p$  & $3/2$ & $5/2$ & $7/2$  \\
        Number of strings improved & 9 & 17 & 28  \\
        Mean cost of improved strings & $11.03$ & $11.23$ & $11.26$\\
        Largest cost of improved strings & $11.05$ & $12.02$ & $12.96$  \\
        \multicolumn{4}{c}{\rule{0pt}{2ex}} \\
    \multicolumn{4}{c}{\textsc{Warm start at $C_{\text{SK}}(w)= 17$}}\\
    \multicolumn{4}{c}{262 classical strings  $w$, 50 samples}\\
    \hline
        $p$  & $3/2$ & $5/2$ & $7/2$   \\
        Number of strings improved & 1 & 4 & 5 \\
        Mean cost of improved strings & 17.003 & 17.18 & 17.28 \\
        Largest cost of improved strings & 17.003 & 17.39 & 17.95  \\ 
    \end{tabular}\\
    \caption{\textit{Sherrington-Kirkpatrick Model. QAOA improvement from good initial classical strings on a 14-vertex instance of the SK model with lowest energy $-33$ and highest energy $+35$. The number of parameters is $2p$. At energy $11$, improvement was possible using magic angles on four classical strings, two of which went to $19$ and two to $27$.}\label{tab:warm_start_sk_14}}
    \begin{tabular}{r|cccc}
    \multicolumn{5}{c}{\rule{0pt}{2ex}} \\
    \multicolumn{5}{c}{\textsc{Standard QAOA}}\\
    \hline
        $p$  & $1$ & $2$ & $3$ & $4$  \\
        Expected cost & 15.05 & 20.40 & 24.27 & 26.19 \\
        \end{tabular}
    \end{center}
    \caption{\textit{Sherrington-Kirkpatrick Model. Performance of the standard QAOA for the same instance as Table~\ref{tab:warm_start_sk_14}, starting in the uniform superposition where the expected value of $C_{\text{SK}}$ is 0. The number of parameters is $2p$. Note the steady improvement with depth. }    \label{tab:standard_sk_14}}
\end{table}

\subsection{Magic angles for the Sherrington-Kirkpatrick model\label{sec:magic_angle}}

We now discuss the exceptions. For the SK model on an even number of sites, with any set of couplings $\{J_{i,j}\}$, there are parameters at $p=3/2$ which creates a cat state starting on any computational basis state. The  ``magic angle'' unitary which creates the cat state is given by:
\begin{align}
    U_{\text{magic}} = e^{-i \frac{\pi}{4} B} e^{i \frac{\pi}{4}  C_{\text{SK}}} e^{i\frac{\pi}{4}  B}.
\end{align}
In this section we show that when $n$ is even, this unitary takes any initial string to a cat state of two strings given (up to global phase) by
\begin{align}
    U_{\text{magic}} \ket{w} = \frac{1}{\sqrt{2}}\left(\ket{w'}+e^{i\theta} \ket{-w'}\right),\label{eq:flipped_w}
\end{align}
where $\theta$ is a known phase, and where 
\begin{align}
    \ket{w'} = \prod_{J_{i, j} = -1}( -iY_i Y_j)\ket{w}. \label{eq:w_prime}
\end{align}
In other words, evolving with the QAOA unitary at these parameters takes $\ket{w}$ to a cat state involving  $\ket{w'}$ and $\ket{-w'}$, where $w'$ is obtained from $w$ by flipping the $i$th bit if the product over $j$ of $J_{i,j}$ is $-1$. For a typical $\{J_{i,j}\}$ this makes it seem that $C(w')$ is essentially random, not different in distribution from $C(w'')$ where $w''$ is chosen at random. If $w$ is a good string, it is unlikely that $w'$ will be better. On rare occasions, this transformation can improve the cost function. However, improvement becomes less likely for better initial strings. We see this in our numerics: when starting from strings with energy $17$, none of the 50 warm start samples improve under the magic angle unitary, whereas four do improve starting at the lower energy of $11$.

To show this result, let us first consider the case where all $J_{i, j}=1$. Using the fact that $e^{i\frac{\pi}{4}X}Ze^{-i\frac{\pi}{4}X} = Y$, we can simplify the unitary to 
\begin{align}
     e^{-i \frac{\pi}{4} B} e^{i\frac{\pi}{4} \sum_{i, j}Z_i Z_j} e^{i\frac{\pi}{4}  B}\ket{w} =  e^{i\frac{\pi}{4} \sum_{i, j}Y_i Y_j}\ket{w} = e^{-i\frac{\pi}{8}n}e^{i\frac{\pi}{8} (\sum_{i}Y_i)^2}\ket{w}.\label{eq:sum_of_ys}
\end{align}

We will evaluate~\eqref{eq:sum_of_ys} in the $y$-basis by writing the $Y_i$ operators in terms of their eigenvalues $y_i$, which are $\pm 1$ valued bits. In particular, we will show by induction that for $n$ even, we have

\begin{align}
    \exp\bigg({i\frac{\pi}{8} \Big(\sum_{i}y_i\Big)^2}\bigg) = \frac{e^{i\frac{\pi}{4}}}{\sqrt{2}} \left(1 - i^{n+1} y_1 y_2 \cdots y_n\right).\label{eq:cat}
\end{align}
The $n=2$ base case is immediate:

\begin{align}
    \exp\left({i\frac{\pi}{8} (y_1 + y_2)^2}\right)  = \frac{e^{i\frac{\pi}{4}}}{\sqrt{2}}(1 - i^3 y_1 y_2).
\end{align}
Now to go from the $n-2$ case to $n$ we have
\begin{align}
    \exp\bigg({i\frac{\pi}{8} \Big(\sum_{i}y_i\Big)^2}\bigg) &= \exp\left({i\frac{\pi}{8}(y_1+\dots + y_{n-2})^2+i\frac{\pi}{4}(y_1+\dots + y_{n-2})(y_{n-1}+y_n)+i\frac{\pi}{8}(y_{n-1}+y_n)^2}\right)\\
     &=\frac{e^{i\frac{\pi}{4}}}{\sqrt{2}} \left(1 - i^{n-1} y_1 y_2 \cdots y_{n-2}\right) \exp\left({i\frac{\pi}{4}(y_1+\dots + y_{n-2})(y_{n-1}+y_n)+i\frac{\pi}{8}(y_{n-1}+y_n)^2}\right). \label{eq:n-2_to_n}
\end{align}
We now split the analysis into two cases: when $y_{n-1}=1$ and $y_n=-1$, and when $y_{n-1}=1$ and $y_n=1.$ For the first case, we can simplify~\eqref{eq:n-2_to_n} to 
\begin{align}
    =\frac{e^{i\frac{\pi}{4}}}{\sqrt{2}}\left(1-i^{n-1}y_1\dots y_{n-2}\right)&=\frac{e^{i\frac{\pi}{4}}}{\sqrt{2}}\left(1+i^{n-1}y_1\dots y_{n}\right)= \frac{e^{i\frac{\pi}{4}}}{\sqrt{2}}\left(1-i^{n+1}y_1\dots y_{n}\right),
\end{align}
which satisfies the inductive hypothesis.
For the second case, we have
\begin{align}
    \frac{e^{i\frac{\pi}{4}}}{\sqrt{2}}(1-i^{n-1}y_1\dots y_{n-2})e^{i\frac{\pi}{2}(y_1 +\dots+ y_{n-2})}e^{i\frac{\pi}{2}}&=
    \frac{e^{i\frac{\pi}{4}}}{\sqrt{2}}(1-i^{n-1}y_1\dots y_{n-2})i^{n-2}(y_1 \dots y_{n-2})i\\
    &=\frac{e^{i\frac{\pi}{4}}}{\sqrt{2}}(1-i^{n+1} y_1 \dots y_{n}),
\end{align}
where in the last line we have used the fact that $n$ is even. This agrees with the inductive hypothesis as well, including all phases, so we have shown that equation~\eqref{eq:cat} is true.

So the action of $U_{\text{magic}}$ on $\ket{w}$ is
\begin{align}
    U_{\text{magic}}\ket{w} = \frac{e^{-i\frac{\pi}{8}(n-2)}}{\sqrt{2}} \left(1 - i^{n+1} Y_1 Y_2 \cdots Y_n\right)\ket{w} .
\end{align}
This operator creates a superposition of $\ket{w}$ and its complement $\ket{-w}$ when all $J_{i,j}=1$. 

Now consider what happens if we change a particular $J_{i, j}$ to $-1$. The new unitary $U_{\text{magic}}'$ is related to $U_{\text{magic}}$ as
\begin{align}
    U_{\text{magic}}' = e^{-i\frac{\pi}{2} Y_iY_j}U_{\text{magic}} = -iY_iY_j U_{\text{magic}},
\end{align}
so the net effect of flipping $J_{i, j}$ is to flip the corresponding bits in the cat state. Flipping additional $J$'s repeats this effect, so we always end up with a state (up to global phase) of the form
\begin{align}
    \frac{1}{\sqrt{2}}\big(\ket{w'} + e^{i\theta}\ket{-w'}\big),
\end{align}
where $w'$ is obtained from the original string $w$ by flipping the sequence of bits associated with the negative $J$'s as described in equation~\eqref{eq:w_prime}.


\section{Condition for MaxCut improvement at small angles for $p=3/2$\label{sec:small_angles_proofs}}
Here we give conditions for the $p=3/2$ QAOA for $\beta_1, \beta_2\ll1 $ to improve upon its initial string. Our final result is stated in equation~\ref{eq:sam_iff} in the main text.

The expected cost produced by the $p=3/2$ warm started QAOA is, for the MaxCut cost $C$,
\begin{align}
    &\langle w|e^{i\beta_1B}e^{i\gamma C}e^{i\beta_2B}Ce^{-i\beta_2B}e^{-i\gamma C}e^{-i\beta_1 B}|w\rangle.
\end{align}
We now expand to order $\beta^2$. The zeroth-order term is $C(w)$ and the higher order terms are the change from this value. Terms of order $\beta$ cancel, and the terms of order $\beta^2$ are
\begin{align}
    &-\beta_1\beta_2\langle w|Be^{i\gamma C}BC e^{-i\gamma C}|w\rangle + \mathrm{c.c.} + \beta_1\beta_2\langle w|Be^{i\gamma C}CBe^{-i\gamma C}|w\rangle + \mathrm{c.c.} \nonumber\\
     &- \frac{\beta_1^2}{2}\langle w|B^2e^{i\gamma C}Ce^{-i\gamma C}|w\rangle + \mathrm{c.c.} - \frac{\beta_2^2}{2}\langle w|e^{i\gamma C}B^2Ce^{-i\gamma C}|w\rangle + \mathrm{c.c.} \nonumber\\ 
     &+ \beta_1^2\langle w|Be^{i\gamma C}Ce^{-i\gamma C}B|w\rangle + \beta_2^2\langle w|e^{i\gamma C}BCBe^{-i\gamma C}|w\rangle \nonumber \\
     \nonumber \\
    &\qquad\qquad\qquad\qquad\qquad = -2\beta_1\beta_2C(w)\sum_i\cos[\gamma (C_i-C(w))] + 2\beta_1\beta_2\sum_i C_i\cos[\gamma(C_i-C(w))] \nonumber \\
    & \qquad\qquad\qquad\qquad\qquad\qquad -\beta_1^2nC(w) - \beta_2^2 n C(w) + \beta_1^2\sum_i C_i + \beta_2^2\sum_i C_i,
\end{align}
where $C_i = \bra{w} X_i C X_i \ket{w}$. We define the quantity
\begin{align}
    \delta_i = (C_i - C(w))/2.
\end{align}
Then we can replace each factor of $n$ with a sum over $i$, and substitute the definition of $\delta_i$ to obtain
\begin{align}
    2(\beta_1^2 + \beta_2^2)\sum_i\delta_i + 4\beta_1\beta_2\sum_i\delta_i\cos2\gamma\delta_i = 2(\beta_1+\beta_2)^2\sum_i\delta_i - 8\beta_1\beta_2\sum_i\delta_i\sin^2\gamma\delta_i.
\end{align}
Therefore we have improvement if we can find parameters for which 
\begin{align}
    2(\beta_1+\beta_2)^2\sum_i\delta_i - 8\beta_1\beta_2\sum_i\delta_i\sin^2\gamma\delta_i >0.\label{eq:sam_3/2} 
\end{align}
 
We restrict our attention to $3$-regular graphs, for which $|\delta_i| = 1$ or $3$.
We want conditions on $\{\delta_i\}$ under which the inequality~\eqref{eq:sam_3/2} can be met for small $\beta_1,\beta_2$, and some $\gamma$.  Let us split the analysis into two cases, depending on whether $\sum_{i}\delta_i\sin^2(\gamma\delta_i)$ is positive or negative. 

If $\sum_{i}\delta_i\sin^2(\gamma\delta_i)$ is positive, then we can always meet~\eqref{eq:sam_3/2}, for example, by taking $\beta_2 = -\beta_1$. So we seek conditions on $\{\delta_i\}$ that guarantee that  $\sum_{i}\delta_i\sin^2(\gamma\delta_i)$ is positive for some $\gamma$. Now,
\begin{align}
    \sum_i\delta_i\sin^2(\gamma\delta_i) &= \sum_{|\delta_i|=1}\delta_i\sin^2\gamma + \sum_{|\delta_i|=3}\delta_i\sin^2(3\gamma)\\
    &=\sum_{|\delta_i|=1}\delta_i\sin^2\gamma + \sum_{|\delta_i|=3}\delta_i(3\sin\gamma - 4\sin^3\gamma)^2\\
    &= \sin^2\gamma \bigg[\sum_{|\delta_i|=1}\delta_i +(3-4\sin^2\gamma)^2\sum_{|\delta_i|=3}\delta_i\bigg].\label{eq:3/2_negative}
\end{align}
We have that $0\le (3-4\sin^2\gamma)^2\le 9$, so we see that the term in brackets ranges from the first sum to the first sum plus $9$ times the second sum as $\gamma$ varies. This means  that~\eqref{eq:3/2_negative} is positive for some $\gamma$ if and only if 
\begin{align} \qquad \sum_{|\delta_i|=1}\delta_i > 0\qquad\text{or}\qquad\sum_{|\delta_i|=1}\delta_i + 9\sum_{|\delta_i|=3}\delta_i >0 .
\end{align}
The second condition is exactly equal to $\sum_i\delta_i^3 >0$.

Therefore  $\sum_{i}\delta_i\sin^2(\gamma\delta_i)$ can be positive if and only if
\begin{align}
 \sum_{|\delta_i|=1}\delta_i >0\qquad\text{or}\qquad \sum_i\delta_i^3 >0. 
\end{align}
Now let us examine the case where $\sum_{i}\delta_i\sin^2(\gamma\delta_i)$ is negative. In this case, we can succeed in satisfying  ~\eqref{eq:sam_3/2}  as long as
\begin{align}
    \sum_i\delta_i > \frac{4\beta_1\beta_2}{(\beta_1+\beta_2)^2}\sum_{i}\delta_i\sin^2(\gamma\delta_i)
\end{align}
for some $\beta_1,\beta_2$. The right-hand side is minimized at $\beta_1=\beta_2$, so this condition can hold if and only if
\begin{align}
    \sum_i\delta_i > \sum_i \delta_i\sin^2(\gamma\delta_i),
\end{align}
or equivalently, if for some $\gamma$,
\begin{align}
    \sum_i\delta_i\cos^2(\gamma\delta_i) > 0.
\end{align}
Repeating now the analysis done for the $\sum_{i}\delta_i\sin^2(\gamma\delta_i) > 0$ case, we obtain the same conditions for success.
In summary, at $p=3/2$ improvement of the cost is possible for small $\beta$'s if and only if 
\begin{align}
 \sum_{|\delta_i|=1}\delta_i >0\qquad\text{or}\qquad \sum_i\delta_i^3 >0
\end{align}
which proves equation~\eqref{eq:sam_iff} in the main text.

\section{Thermality argument proofs\label{sec:thermal_proofs}}

\subsection{Proof that typical strings cannot be improved\label{subsec:average_string_failure}}
 In this section we will show for any $\epsilon>0$, with $w$ chosen uniformly at random,
\begin{align}
    \Pr_w\left(\frac{1}{m}\langle w|U_w^\dagger CU_w |w\rangle - \frac{1}{m}C(w) \ge \frac{1}{m^{1/2-\varepsilon}}\right) \rightarrow 0,\qquad\text{as }m\rightarrow\infty,
\end{align}
where $U_w$ is the optimal QAOA unitary for each $w$. 

To streamline the presentation, we restrict ourselves to the $p=3/2$ QAOA on triangle-free $d$-regular graphs, although the argument can be generalized to any constant depth, and to general bounded-degree graphs.  Let $G = (V,E)$ be a triangle-free $d$-regular graph with $|V|=n$ vertices and \mbox{$|E|=m=nd/2$} edges.  We consider cost functions of the generic form 
\begin{align}
        C(z) = \sum_{\langle i, j\rangle\in E}C_{\langle i, j\rangle}(z_i, z_j).
\end{align}
We can use the observations of Section~\ref{subsec:nbhd_sums} to reorganize the QAOA expectation: at $p=3/2$, the QAOA involves only a small tree around each edge consisting of $2d$ vertices.
That is, we can write the expected value of the cost operator (normalized by the number of edges) as
\begin{align}
    \frac{1}{m}\langle w|U^\dagger_w C U_w|w\rangle &= \frac{1}{m}\sum_{\langle i,j\rangle \in E}\langle w|U_w^\dagger C_{\langle i,j\rangle}U_w|w\rangle,\label{eq:cost_edge_decomp}
\end{align}
where each matrix element on the right-hand side only depends on the $2d$ bits of $w$ which lie on the tree whose central edge is $\langle i,j\rangle$. We call the restriction of $w$ to these bits $w_{\langle i,j\rangle}$. So we can write
\begin{align}
\frac{1}{m}\langle w|U^\dagger_w C U_w|w\rangle &=\frac{1}{m}\sum_{\langle i,j\rangle \in E}\langle w_{\langle i,j\rangle}|U_w^\dagger C_{\langle i,j\rangle}U_w|w_{\langle i,j\rangle}\rangle,
\end{align}
where each $w_{\langle i,j\rangle}$ is a bit string on $2d$ bits. Summing over the possible values of such bit strings, we can rearrange the sum as
\begin{align} \label{eq:edge_to_nbhd}
    \frac{1}{m}\langle w|U^\dagger_w C U_w|w\rangle &= \frac{1}{m}\sum_{\ell \in \{-1,1\}^{2d}}\#_w(\ell)\langle \ell|U_w^\dagger C_\mathrm{edge} U_w|\ell\rangle,
\end{align}
where $\#_w(\ell)$ is the number of times the $2d$-bit string $\ell$ appears within the collection of $w_{\langle i,j\rangle}$, and where $C_{\mathrm{edge}}$ denotes a generic edge term (for MaxCut, $C_{\mathrm{edge}}=-Z_a Z_b$, where $\langle a,b\rangle$ labels the central edge of a tree neighborhood).  We consider $U^\dagger_wC_{\mathrm{edge}}U_w$ as an operator on the $2d$-qubit Hilbert space associated with the $2d$ vertices in the radius-1 tree neighborhood of an edge. 

Now we consider the case where the string $w\in \{-1,1\}^{n}$ is chosen uniformly at random. We can think of $\#_w(\ell)$ as a random variable whose average value is $m/2^{2d}$. We conveniently partition the sum as
\begin{align} \label{eq:correction}
    \frac{1}{m}\langle w|U^\dagger_w C U_w|w\rangle &= \frac{1}{2^{2d}}\sum_{\ell \in \{-1,1\}^{2d}}\langle \ell|U_w^\dagger C_\mathrm{edge} U_w|\ell\rangle + \frac{1}{m}\sum_{\ell \in \{-1,1\}^{2d}}\left(\#_w(\ell) - \frac{m}{2^{2d}}\right)\langle \ell|U_w^\dagger C_\mathrm{edge} U_w|\ell\rangle.
\end{align}
The first term of the right-hand side of equation~\eqref{eq:correction} turns out to be a simple constant,

\begin{align}
     \frac{1}{2^{2d}} \sum_{\ell \in \{-1,1\}^{2d}}\langle \ell|U_w^\dagger C_\mathrm{edge} U_w|\ell\rangle &= 
      \frac{1}{2^{2d}} \sum_{\ell \in \{-1,1\}^{2d}}\mathrm{Tr}_{\textrm{tree}}\left(U_w^\dagger C_\mathrm{edge} U_w |\ell\rangle\langle \ell|\right)\\
    &=  \frac{1}{2^{2d}}\mathrm{Tr}_{\textrm{tree}}\left(U_w^\dagger C_\mathrm{edge} U_w I\right)\\
    &= \frac{1}{2^{2d}} \mathrm{Tr}_{\textrm{tree}}(C_{\mathrm{edge}})\\
    &= \frac{1}{ 2^n}  \sum_{z \in \{-1,1\}^n } \frac{1}{m} \langle z | C | z \rangle  \\
    &= \bar{c},
\end{align}
where $\overline{c}$ is the average cost value per edge as defined in equation~\eqref{eq:avg_cost_per_edge}, and where $\Tr_{\textrm{tree}}$ refers to the trace over the Hilbert space of an abstract tree neighborhood.

The second term in \eqref{eq:correction} is then the deviation of the optimized cost function $\frac{1}{m}\langle w|U^\dagger_w C U_w|w\rangle$ from the average cost function $\bar{c}$.  We want to upper bound this deviation, which we denote by $\chi$, defined as
\begin{align} \label{eq:chi_def}
    \chi & =  \frac{1}{m}\langle w|U^\dagger_w C U_w|w\rangle - \bar{c} \\  
    & = \frac{1}{m}\sum_{\ell \in \{-1,1\}^{2d}}\left(\#_w(\ell) - \frac{m}{2^{2d}}\right)\langle \ell|U^\dagger_w C_\mathrm{edge}U_w|\ell\rangle.
\end{align}
We can bound $\chi^2$ using the Cauchy-Schwarz inequality as
\begin{align}
    \chi^2 &\le \frac{1}{m^2}\bigg[\sum_{\ell \in \{-1,1\}^{2d}}\left(\#_w(\ell) - \frac{m}{2^{2d}}\right)^2\bigg]\bigg[\sum_{\ell \in \{-1,1\}^{2d}}\left|\langle \ell| U_w^\dagger C_{\mathrm{edge}} U_w|\ell\rangle\right|^2\bigg]\\
    &\le \frac{1}{m^2}2^{2d}\sum_{\ell \in \{-1,1\}^{2d}} \left(\#_w(\ell) - \frac{m}{2^{2d}}\right)^2,
\end{align}
where we use the fact that $|C_\mathrm{edge}|\le 1$ in the second inequality. For each $\ell \in \{-1,1\}^{2d}$, we can imagine $\#_w(\ell)$ as a sum of indicator random variables $B_{\langle i,j\rangle}$. Each $B_{\langle i,j\rangle}$ is $1$ if $w_{\langle i,j\rangle} = \ell$, where we recall that $w_{\langle i,j\rangle}$ is the restriction of $w$ to the edge neighborhood centered at edge $\langle i,j\rangle$, and $0$ otherwise: 
\begin{align}
    \#_w(\ell) = \sum_{\langle i,j\rangle \in E}B_{\langle i,j\rangle}.
\end{align}
The mean of each $B_{\langle i,j\rangle}$ is $2^{-2d}$ because the string $w$ is chosen uniformly at random. Note $B_{\langle i,j\rangle}$ and $B_{\langle i',j'\rangle}$ are independent when the tree neighborhoods of the corresponding edges do not overlap. For each $B_{\langle i,j\rangle}$, there are at most $D = 2d^3$ distinct edge neighborhoods with which it has non-trivial overlap. Averaging over the $w$ we get
\begin{align}
    \Ex_w\left[\left(\#_w(\ell) - \frac{m}{2^{2d}}\right)^2\right] &= \sum_{\langle i,j\rangle \in E}\sum_{\langle i',j'\rangle \in E}\Ex_w\left[\left(B_{\langle i,j\rangle} - \frac{1}{2^{2d}}\right)\left(B_{\langle i',j'\rangle} - \frac{1}{2^{2d}}\right)\right]\\
    &\le mD\frac{1}{2^{2d}}\left(1 - \frac{1}{2^{2d}}\right),
\end{align}
because there are at most $mD$ pairs of overlapping neighborhoods in the double sum (the other terms vanish), and for these nonzero terms, by Cauchy-Schwarz, the covariance $\Ex_w\left[\cdots\right]$ is upper bounded by the variance $\mu(1-\mu)$ of the indicator random variables $B_{\langle i,j\rangle}$ with mean $\mu=2^{-2d}$.

Altogether, $\Ex_w[\chi^2]$ is bounded above as
\begin{align}
    \Ex_w\left[\chi^2\right] &\le \frac{1}{m^2}2^{2d}mD\frac{1}{2^{2d}}\left(1 - \frac{1}{2^{2d}}\right) \le \frac{D}{m}.
\end{align}

By the Chebyshev inequality, for any $a > 0$, we have
\begin{align}
    \Pr_w\left(\left|\chi\right| \ge a\right) \le \frac{\Ex_w[\chi^2]}{a^2} \le \frac{D}{ma^2}.
\end{align}
Choosing $a=m^{-1/2+\varepsilon}$ for any $\varepsilon > 0$, and plugging in the definition of $\chi$ in equation~\eqref{eq:chi_def}, we get
\begin{align} 
    \Pr_w\left(\left|\frac{1}{m} \langle w|U^\dagger_w C U_w|w\rangle - \bar{c} \right| \ge \frac{1}{m^{1/2 - \varepsilon}}\right) \longrightarrow 0,\quad\text{as }m\rightarrow\infty..\label{eq:chi_zero_app}
\end{align}
It follows that as the number of edges $m$ tends to infinity, the probability of any non-negligible deviation from $\bar{c}$ tends to zero. Equation~\eqref{eq:chi_zero_app}, combined with the law of large numbers (equation~\eqref{eq:controverial_equation}, main text), finally yields the bound
\begin{align}
    \Pr_w\left(\frac{1}{m}\langle w|U_w^\dagger CU_w |w\rangle - \frac{1}{m}C(w) \ge \frac{1}{m^{1/2-\varepsilon}}\right) \rightarrow 0,\qquad\text{as }m\rightarrow\infty.
\end{align}
We conclude that for nearly all initial strings, the constant-depth QAOA can only improve the cost function (per edge) by an amount that vanishes in the limit of large graphs.

\subsection{Proof that good strings cannot be improved}\label{subsection:thermal_proof}
In this section we develop a statistical argument to upper bound the progress made by the QAOA starting from a good string.  The bound applies to the QAOA at constant depth on large bounded-degree graphs.  In the previous section, we proved the QAOA makes little progress when starting from a string chosen uniformly at random.  Because such strings are not generally ``good'' strings, we extend the argument to good strings we describe as ``locally thermal.''  

\subsubsection{Intuition for the thermal argument}
We begin with an intuitive explanation of our argument, then state the main Theorem~\ref{thm:cbar0_app} and its derivation.
Our argument relies on the locality and the uniformity of the operators appearing in the QAOA circuit. Recall that we consider cost functions that are a sum of terms for each edge of a graph, as in equation~\eqref{eq:cost_function}. Let $U$ be the QAOA operator from equation~\eqref{eq:warm_start_QAOA_unitary} at depth $p$. The QAOA operator for a given edge is $U^\dagger C_{\langle i,j\rangle} U$, which acts non-trivially only on the edge neighborhood of $\langle i,j \rangle$ with maximal radius $r=p-1/2$ (see Section~\ref{subsec:nbhd_sums}).  By uniformity, we mean that the reduced operator on any subgraph depends only on the isomorphism type of the underlying subgraph (i.e., the reduced operators on different subgraphs are the same when the subgraphs are isomorphic). Note that the operators of the standard QAOA satisfy this uniformity property as long as the $C_{\langle i,j\rangle}$ do not depend on the edge $\langle i,j \rangle$.  Our statistical argument also relies on bounded degree, so it does not apply directly to problems like the Sherrington-Kirpatrick model in Appendix~\ref{sec:SK}.

For MaxCut on $d$-regular graphs as $n\to\infty$ and at constant depth, all but a vanishing fraction of edge neighborhoods are trees. Following the main text, we define the fraction $\delta$ of edges whose neighborhoods are \textit{not} trees as
\begin{align}  \label{eq:tree_frac}
\delta = 1-\frac{|E_T|}{m} ,
\end{align}
where $E_T$ denotes the set of edges whose neighborhood of radius $r$ is a tree. 

Recall the ensemble $\rhowtree$ associated with the warm-start string $w$. For any classical string \mbox{$w\in \{-1,1\}^n$} and a mostly locally tree-like $G$ we define its \emph{local ensemble} by
\begin{align} \label{eq:local_ensemble}
    \rhowtree = \frac{1}{|E_T|}\sum_{\langle i,j\rangle \in E_T}|w_{\langle i,j\rangle}\rangle\langle w_{\langle i,j\rangle}|,
\end{align}
where $w_{\langle i,j\rangle}$ denotes the restriction of $w$ onto the local tree neighborhood $T_{\langle i,j\rangle}$, and  $T_{\langle i,j\rangle}$ denotes the tree neighborhood centered around edge $\langle i,j\rangle \in E_T$. Each tree neighborhood has a fixed number of vertices, so $w_{\langle i,j \rangle}$ is a string of that length, and $\rhowtree$ is a density matrix over the corresponding number of qubits.

We now use the global thermal density matrix $\rho_\beta$ with respect to $H=-C$,  (see equation~\ref{eq:rhobeta_defn} in the main text) to define the reduced density matrix onto the tree neighborhood $T_{\langle i,j\rangle}$, denoted $\rhobetaij$, and the ensemble
\begin{align} \label{eq:rhobetatree_app}
    \rhobetatree = \frac{1}{|E_T|}\sum_{\langle i,j\rangle \in E_T} \rhobetaij.
\end{align}
Again, both $\rhowtree$ and $\rhobetatree$ act on the Hilbert space of a single tree neighborhood. 

Consider two strings  $w$ and $w'$ drawn from the thermal ensemble $\rho_\beta$.  While the strings may look totally different on any particular neighborhood, you might expect that $\rhowtree$ and $\rho_{w',\textrm{tree}}$ are nonetheless similar: both strings look the same locally when averaged over neighborhoods. 
In that case, typical strings from the thermal ensemble should all have the same expected cost function under $U$, since the expected cost is completely determined by $\rhowtree$ up to small corrections in $\delta$. In particular, the cost function should behave as if the thermal state $\rho_\beta$ were itself the input to the QAOA. But a key property of the thermal state is that it has the minimum expected energy of all states with the same entropy.  That is, fixing some $\beta>0$, we have
\begin{align} \label{eq:min_energy_fixed_entropy}
    \rho_\beta = \argmin_{\rho \, :\,  S(\rho) = S(\rho_\beta) } \mathrm{Tr}(H\rho).
\end{align}
In particular, since unitaries preserve entropy, we have
\begin{align}  \label{eq:min_energy}
    \min_{U}\mathrm{Tr}( HU \rho_\beta U^\dagger) = \mathrm{Tr}(H\rho_\beta).
\end{align} 
This means that no unitary applied to $\rho_\beta$ can lower the energy.  

Suppose $\rhowtree$ is well-approximated by the density matrix  $\rhobetatree$.  
Then $w$ gives the same expectation for the cost-function as the associated $\rho_\beta$ under any QAOA unitary, so no such unitary applied to $w$ can reduce the energy. We call such strings $w$ ``locally thermal,'' and we quantify this property with the ``thermality coefficient'' $\varepsilon_w$, 
\begin{align}
\label{eq:local_therm_app}
    \varepsilon_w =  \| \rhobetatree- \rhowtree\|_1,
\end{align}
where $\beta$ is the inverse temperature chosen so that 
\begin{align}
\mathrm{Tr}(C\rho_\beta) = C(w).
\end{align} 
Note that $\varepsilon_w$ depends on the QAOA depth $p$ which defines the neighborhood radius. Our main result, stated next (Theorem~\ref{thm:cbar0} in the main text), shows that if $w$ is locally thermal, i.e., if $\varepsilon_w$ is small, then $w$ can only improve by a small amount. 
Locally thermal strings may be viewed as exhibiting a form of ergodicity: the average over local neighborhoods for the fixed string looks like the thermal average over different global strings.

\subsubsection{Proof of the thermal argument}
We now prove our main result, which bounds the improvement of the single-string warm-start QAOA for locally thermal strings.
\begin{theorem}\label{thm:cbar0_app}
Consider the single-string warm-start QAOA initialized in a string $w$ with thermality coefficient $\varepsilon_w$, and let $\delta$ be the fraction of neighborhoods whose edges are not trees. Then
    \begin{align} 
     \frac{1}{m}\langle w|U^\dagger_w C U_w|w\rangle  \leq c(w) +2\varepsilon_w + 4\delta ,
\end{align}
where $c(w)= C(w)/m$ is the cut fraction of the original string $w$.
\end{theorem}

Recalling equation~\eqref{eq:cost_edge_decomp} and the observations of Section~\ref{subsec:nbhd_sums}, we have

\begin{align}
    \frac{1}{m}\langle w|U^\dagger_w C U_w|w\rangle 
    &= \frac{1}{m}\sum_{\langle i,j\rangle \in E}\langle w|U_w^\dagger C_{\langle i,j\rangle}U_w|w\rangle\\
    &\leq \frac{1}{m} \sum_{\langle i,j\rangle \in E_T}\langle w|U_w^\dagger C_{\langle i,j\rangle}U_w|w\rangle + \delta \\
    &= \frac{1}{m}\sum_{\langle i,j\rangle \in E_T}\langle w_{\langle i,j\rangle}|U_w^\dagger C_{\langle i,j\rangle}U_w|w_{\langle i,j\rangle}\rangle + \delta \\
    &= \frac{1}{m}\Tr_{\textrm{tree}}\big(U_w^\dagger C_{\mathrm{edge}}U_w \sum_{\langle i,j\rangle \in E_T}|w_{\langle i,j\rangle}\rangle\langle w_{\langle i,j\rangle}|\big) + \delta \\
    &=\frac{|E_T|}{m} \Tr_{\textrm{tree}}\big(U_w^\dagger C_{\mathrm{edge}}U_w\rhowtree\big) + \delta .
\end{align}
The second line drops the contribution from edges whose neighborhoods are not trees.  The $\Tr_{\textrm{tree}}$ refers to the trace over the Hilbert space associated to an abstract tree neighborhood. The final line uses the definition of $\rhowtree$ in~\eqref{eq:local_ensemble}.

Next, we replace $\rhowtree$ with $\rhobetatree$, picking up an error proportional to $\varepsilon_w$, which measures the local thermality of $w$.  We use the property of the trace norm that
 \begin{align}
     \Tr(A^\dagger B) \le \|A\|_\infty\cdot \|B\|_1,
 \end{align}
where $\norm{A}_\infty$ is the operator norm (the largest singular value). Recall that we assume $C$ is normalized so that $|C_{\langle i, j \rangle}| \leq 1$, or equivalently, $||C_{\textrm{edge}}||_{\infty} \leq 1$. Then we find
\begin{align}
    \frac{1}{m}\langle w|U^\dagger_w C U_w|w\rangle \leq \frac{|E_T|}{m} \Tr_{\textrm{tree}}\left(U_w^\dagger C_{\mathrm{edge}}U_w
    \rhobetatree 
    \right) + \varepsilon_w  + \delta.
\end{align}

Now we substitute \eqref{eq:rhobetatree_app}, replace the partial trace $\mathrm{Tr}_{G\backslash T_{\langle i,j\rangle}}(\rho_{\beta})$ with the full thermal operator $\rho_\beta$, and re-sum the cost Hamiltonian $C$:
\begin{align}
    \frac{1}{m}\langle w|U^\dagger_w C U_w|w\rangle 
    &\leq \frac{1}{m}\sum_{\langle i,j\rangle\in E_T}\Tr_{\textrm{tree}}\left(U_w^\dagger C_{\langle i,j\rangle}U_w\mathrm{Tr}_{G\backslash T_{\langle i,j\rangle}}(\rho_{\beta})\right) + \varepsilon_w + \delta\\
    & \leq \frac{1}{m}\sum_{\langle i,j\rangle\in E}\mathrm{Tr}\left(U_w^\dagger C_{\langle i,j\rangle}U_w\rho_\beta\right) + \varepsilon_w + 2\delta\\
    &=\frac{1}{m}\mathrm{Tr}\left(U_w^\dagger CU_w\rho_\beta\right) + \varepsilon_w + 2\delta.
\end{align}
Finally, we apply the principle of minimum energy in thermodynamics (equation~\eqref{eq:min_energy}) to the Hamiltonian $H=-C$ with the thermal state $\rho_\beta$ of \eqref{eq:rhobeta_defn}, we get
\begin{align}
    \frac{1}{m}\langle w|U^\dagger_w C U_w|w\rangle & \leq \frac{1}{m}\mathrm{Tr}\left(U_w^\dagger CU_w\rho_\beta\right) + \varepsilon_w + 2\delta \\
    & \leq \frac{1}{m}\mathrm{Tr}(C\rho_\beta) +\varepsilon_w + 2\delta \\
    & \leq \frac{|E_T|}{m}\mathrm{Tr}_{\textrm{tree}}(C_{\textrm{edge}}\rhobetatree) +\varepsilon_w + 3\delta \\
    & \leq \frac{|E_T|}{m}\mathrm{Tr}_{\textrm{tree}}(C_{\textrm{edge}}\rhowtree) +2\varepsilon_w + 3\delta \\
    & \leq  \frac{1}{m}\langle w|C|w\rangle  +2\varepsilon_w + 4\delta 
\end{align}
or equivalently,
\begin{align} 
     \frac{1}{m}\langle w|U^\dagger_w C U_w|w\rangle  \leq c(w) +2\varepsilon_w + 4\delta 
\end{align}
where $c(w)$ is the cut fraction of the original string $w$, reproducing the desired result of Theorem~\eqref{thm:cbar0_app}. Since $\delta=O(1/n)$ is vanishingly small for large graphs, we conclude that the improvement is controlled by $\varepsilon_w$, which measures how close $w$ is to being locally thermal.

\subsection{When are initial strings locally thermal?} \label{sec:why_small}

Here we ask about the values of $\varepsilon_w$ for typical good strings $w$ drawn from the thermal ensemble: how close is $\rhowtree$ to $\rhobetatree$? While $\rhowtree$ is not immediately guaranteed to be close to $\rhobetatree$ for all $w$, the relation holds on average over thermally sampled strings $w$, i.e.,
\begin{align} \label{eq:sampledrhobeta}
    \operatorname*{\mathbb{E}}_{w \sim \rho_\beta} [\rhowtree] = \rhobetatree,
\end{align} where $\operatorname*{\mathbb{E}}_{w\sim\rho_\beta}$ denotes expectation with respect to samples of $w$ from the thermal distribution.  This follows by commuting the thermal average with the average over neighborhoods:
\begin{align}
    \mathop{\mathbb{E}}_{w \sim \rho_\beta} [\rhowtree]  & = \frac{1}{Z} \sum_w e^{-\beta(-C(w))} \rhowtree  \\
    & =  \frac{1}{Z} \sum_w e^{-\beta(-C(w))} \frac{1}{|E_T|} \sum_{\langle i,j\rangle \in E_T}|w_{\langle i,j\rangle}\rangle\langle w_{\langle i,j\rangle}| \\
    & = \frac{1}{|E_T|} \sum_{\langle i,j\rangle \in E_T}  \frac{1}{Z} \sum_w e^{-\beta(-C(w))} |w_{\langle i,j\rangle}\rangle\langle w_{\langle i,j\rangle}| \\
    & =  \frac{1}{|E_T|}\sum_{\langle i,j\rangle \in E_T} \rhobetaij \\
    & = \rhobetatree,
\end{align}
where $Z=\Tr\big(e^{-\beta (-C)}\big)$. In going to the second and fourth lines above, we use the definitions of $\rhowtree$ and $\rhobetatree$ in~\eqref{eq:local_ensemble} and~\eqref{eq:rhobetatree_app}, respectively. 

Since the average of $\rhowtree$ over $w$ matches $\rhobetatree$, to show that $\rhowtree$ is close to $\rhobetatree$ for typical $w$, it suffices to show that $\rhowtree$ has small variance with respect to choice of $w$. Treating $w$ as a random variable distributed according to the thermal ensemble, $\rhowtree$ defined in \eqref{eq:local_ensemble} becomes an average over random variables $|w_{\langle i,j\rangle}\rangle\langle w_{\langle i,j\rangle}|$ defined on each neighborhood.  If these variables have small correlation across different neighborhoods, then $\rhowtree$ has small variance. So $\varepsilon_w$ is small for typical strings sampled from the thermal distribution, assuming the thermal distribution has only small spatial correlations across the graph. At fixed temperature, we therefore expect the variance of $\rhowtree$ shrinks with graph size as $m^{-1}$, and $\varepsilon_w \propto m^{-1/2}$.  

When is it reasonable to assume that $\rho_\beta$ has small spatial correlations across the graph?  At infinite temperature, the case discussed in Section~\ref{sec:thermal_1}, the correlation is zero, so there $\varepsilon_w \propto m^{-1/2}$, consistent with our conclusion there. In the high temperature regime (where the strings are ``good,'' i.e. better than random guesses, but not too good), we similarly expect $\rho_\beta$ to have small correlations.  
At low temperatures, depending on the specific combinatorial problem, there may be thermal phase transitions and development of long-range correlations.  However, since we are interested in warm starts which can be efficiently generated by classical algorithms, this low temperature regime may be less relevant in practice.

\subsection{Numerical evidence for the thermality of typical good strings\label{sec:thermal_num}}

In Section~\ref{subsection:thermal_proof}, we proved that when using a warm start with a classical string $w$, the cost function cannot improve much when $w$ is ``locally thermal.'' More precisely, we showed the improvement of the cost function per edge is controlled by thermality coefficient
\begin{align}
    \varepsilon_w =  \| \rhobetatree- \rhowtree\|_1.
\end{align}
When $w$ is locally thermal, meaning  $\varepsilon_w$ is small,  the average of $w$ over local neighborhoods looks like the thermal average over different global strings. In this section, we present numerical evidence that, for a given cost function, typical good strings $w$  are indeed approximately locally thermal.  Moreover, it appears that the magnitude of the thermality coefficient diminishes with larger graph sizes.  In our investigations below, we only perform numerical experiments for MaxCut, and we only analyze typical good strings associated with one fixed temperature. We leave the analysis of other problems besides MaxCut to future investigations.

We study MaxCut on  random 3-regular graphs of $n=10000$ to $50000$ vertices. The local thermality of typical strings is expected to hold with high accuracy only in the limit of large graphs because it involves the concentration of averages over neighborhoods. Fixing a temperature $T=1.75$, corresponding to a cut fraction of $\sim 0.64$, we use simulated annealing to sample classical strings $w$ from the thermal distribution $\rho_\beta$ of \eqref{eq:rhobeta_defn}.  That is, we sample strings $w$ with probabilities proportional to $ e^{-\beta(-C(w))}$.  

We consider $k=10$ samples, $\{w_\alpha\}$ for $\alpha=1,\ldots,k$.  For each string $w_\alpha$, it is straightforward to compute the associated density matrix $\rho_{w_\alpha,\textrm{tree}}$ of $\eqref{eq:local_ensemble}$.  We consider tree neighborhoods of radius $r=2$, i.e., all vertices within distance $2$ of a given edge  $\langle i,j\rangle$.  For $d=3$, these have 14 vertices, so $\rho_{w_\alpha,\textrm{tree}}$ lives on a Hilbert space of 14 qubits.  

To verify that the sampled $w_\alpha$ are locally thermal, ideally we would calculate $\norm{\rho_{w_\alpha,\textrm{tree}} - \rhobetatree}_1 = \varepsilon_w$.  Unfortunately, for these large graphs with $\geq 10000$ vertices, direct computation of $\rho_\beta$ and  $\rhobetatree$ is impossible in practice.  To proceed, we restate \eqref{eq:sampledrhobeta}, which says that
\begin{align} \label{eq:sampledrhobeta2}
    \operatorname*{\mathbb{E}}_{w \sim \rho_\beta} [\rhowtree] = \rhobetatree,
\end{align} where the  expectation is with respect to samples of $w$ from the thermal distribution. That is, the average of $\rhowtree$ over thermal samples matches $\rhobetatree$. In fact we find that for our numerically generated samples $w_\alpha$, all of the associated $\rho_{w_\alpha,\textrm{tree}}$ are approximately equal.  Thus even with our few $k=10$ samples, the LHS of \eqref{eq:sampledrhobeta2} is already well-approximated by this empirical sample average, and we consider the average deviation
\begin{align} \label{eq:avg_deviation}
    \mathcal{E} = \frac{1}{k} \sum_{\alpha=1}^k \norm{\rho_{w_\alpha,\textrm{tree}} - \overline{\rho_{w,\textrm{tree}}}}_1,
\end{align}
where 
\begin{align}
    \overline{\rho_{w,\textrm{tree}}} = \frac{1}{k} \sum_{\alpha=1}^k \rho_{w_\alpha,\textrm{tree}} .
\end{align}
The value of $\mathcal{E}$ measures the average deviation of $\rho_{w_\alpha,\textrm{tree}}$ over samples.  When $\mathcal{E}$ is small, it follows that typical samples $w$ are approximately locally thermal.  We plot the average deviation $\mathcal{E}$ as a function of the number of vertices $n$ in Figure \ref{fig:nbhd_counting}, with $10$ samples for each $n$. The fit line shows a fit with $\mathcal{E} \propto n^{-1/2}$.  The figure suggests that typical strings $w$ at this temperature $T=1.75$ are approximately locally thermal with $\varepsilon_w$ shrinking as $n^{-1/2}$. This behavior matches the estimate in Section~\ref{sec:why_small}.

\begin{figure}
\centering
\includegraphics[width=0.5\linewidth]{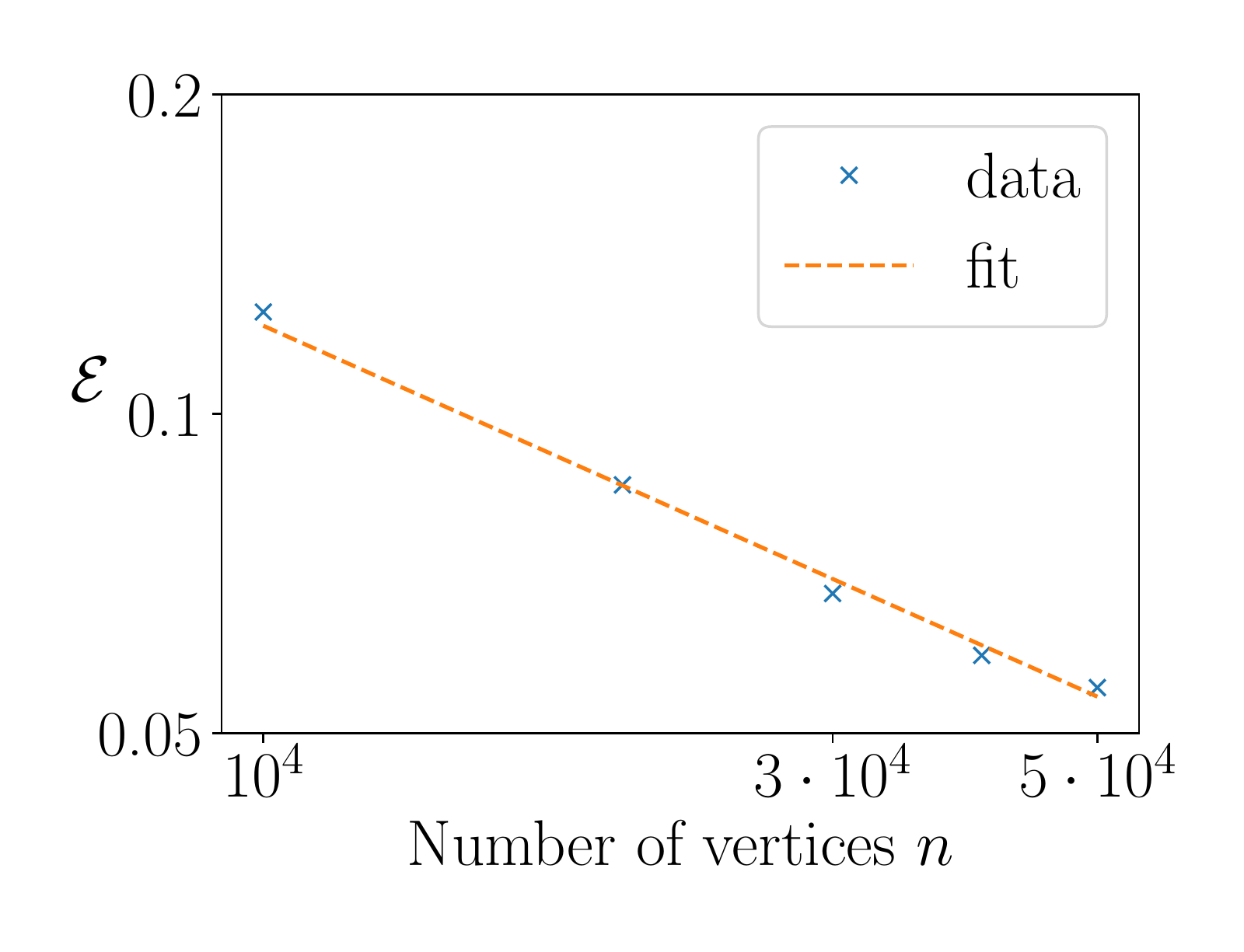}
\caption{\textit{Log-log plot of the average sample deviation $\mathcal{E}$ from \eqref{eq:avg_deviation} as a function of the number of vertices $n$. The fit line shows a one-parameter fit with $\mathcal{E} = c n^{-1/2}$ for the parameter $c$.}} 
\label{fig:nbhd_counting}
\end{figure}

\section{Compression argument proofs}\label{sec:compression_proofs}
We aim to prove Theorem \ref{thm:improvable_main}. We begin with some preliminaries.
\subsection{Preliminary lemmas} \label{sec:nets}
First we review some general ideas about $\epsilon$-nets and $\epsilon$-packings on metric spaces.

Given a set $S$ on a metric space, a subset $S_c \subset S$ is called an \textbf{$\epsilon$-covering} when every point in $S$ is within distance $\epsilon$ of some point in $S_c$. That is, the $\epsilon$-balls around each point in $S_c$ together cover $S$.

A subset $S_p \subset S$ is called an \textbf{$\epsilon$-packing} when all points in $S_p$ are distance at least $\epsilon$ apart. That is, the $\epsilon/2$-balls around each point in $S_p$ are disjoint.

A subset which is both an $\epsilon$-net and an $\epsilon$-packing is called an \textbf{$\epsilon$-net}. We first collect a few useful facts about $\varepsilon$-nets.
\begin{lemma}
Every finite set has at least one $\epsilon$-net, for any $\epsilon>0$.
\end{lemma}

\begin{proof}
To construct the $\epsilon$-net, choose any point in the set and put it in the net.  Throw away all points within an $\epsilon$-ball of the chosen point. If there are no remaining points then we are done. Otherwise, pick a remaining point, put it into the net, and throw away all the points within an $\epsilon$-ball of the point. Continue this process until the set is exhausted. Since the set is finite, this process must terminate.
\end{proof}

\begin{lemma}
A $2\epsilon$-packing cannot have larger cardinality than an $\epsilon$-covering.
\end{lemma}
\begin{proof}
We show, if $S$ has an $\epsilon$-covering $S_c \subset S$ and a $2\epsilon$-packing $S_p \subset S$, then $|S_p| \leq |S_c|$.  To see this, note that all the points in $S_p$ must be within distance $\epsilon$ from some point in $S_c$.  And by triangle inequality, no two distinct points in $S_p$ can be within distance $\epsilon$ of the same point in $S_c$ (otherwise they would be within $2\epsilon$ of each other).  So there's a surjection from $S_p$ to $S_c$, i.e.\ every point in the packing must be covered by a different point in $S_c$.  
\end{proof}

We can show that there is an $\epsilon$-covering of the space of depth-$p$ QAOA unitaries with the following parameters.
\begin{lemma}\label{lem:QAOA_cover}
For any $\delta \in (0,1)$, the set of QAOA unitaries produced by using $2p$ parameters has an $\epsilon$-covering of size $\left(\frac{2\pi}{\delta}\right)^{2p}$ for $\epsilon = p \delta n^2$, using the operator norm distance.
\end{lemma}
\begin{proof}
Consider the set $S$ of QAOA unitaries produced by using $2p$ parameters, each in $[0,2\pi]$.  Let $S_0 \subset S$ be the set of QAOA unitaries obtained by using discretized parameters $\{0,\delta,2\delta,\ldots,2\pi\}$.  Then 
\begin{align} \label{eq:S0_size}
    |S_0| = \left(\frac{2\pi}{\delta}\right)^{2p}.
\end{align}
We want to show that any QAOA unitary is close (in operator norm) to one of these unitaries with discretized parameters.  That is, every point in $S$ is close to some point in $S_0$.

Consider the rotation $e^{i \gamma C}$ that occurs in a single layer of the QAOA circuit, for two different angles $\gamma_1, \gamma_2$, for cost operator $C$.  Then in operator norm
\begin{align}
\norm{e^{i \gamma_1 C}-e^{i \gamma_2 C}} \leq |\gamma_1 - \gamma_2| \norm{C},
\end{align}
where $\norm{C}$ is then the maximum cost of any string. This inequality can be shown by working in the computational basis. Likewise, 
\begin{align}
\norm{e^{i \beta_1 B}-e^{i \beta_2 B}} \leq |\beta_1-\beta_2| \norm{B} \leq |\beta_1-\beta_2|n,
\end{align}
where $n$ is the number of qubits.  

Assuming the clauses in the cost-function are normalized and there are no more than $n^2$, we conclude
\begin{align} \label{eq:rotation-error}
    \norm{e^{i \gamma_1 C}-e^{i \gamma_2 C}} &\leq |\gamma_1 - \gamma_2| n^2, \\
    \norm{e^{i \beta_1 B}-e^{i \beta_2 B}} &\leq |\beta_1 - \beta_2| n .
\end{align}
Now note that for a general product of unitaries, $U_1,\cdots,U_m$ and $V_1,\cdots,V_m$, we have the identity
\begin{align} \label{eq:unitary-error-addition}
    \norm{U_1 \ldots U_m - V_1 \ldots V_m} \leq \sum_{i=1}^m \norm{U_i-V_i}.
\end{align}
If desired one can prove this after observing the $m=2$ case: 
\begin{align}
\norm{U_1 U_2 - V_1 V_2}&=\norm{U_1 (U_2-V_2) + U_1 V_2 - V_1 V_2}\\
&\leq \norm{U_1 (U_2-V_2)} + \norm{U_1 V_2 - V_1 V_2}\\ 
&= \norm{(U_2-V_2)} + \norm{(U_1-V_1) V_2 + V_1 V_2 - V_1 V_2}\\ 
&= \norm{(U_2-V_2)} + \norm{(U_1-V_1)}.
\end{align}

Now consider two QAOA unitaries $U, V$ for which each of their $2p$ parameters differ by at most $\delta/2$.   Then combining the general identity \eqref{eq:unitary-error-addition} with \eqref{eq:rotation-error} we obtain
\begin{align}
    \norm{U-V} \leq p \delta n^2.
\end{align}
So for any QAOA unitary $U \in S$, there is a QAOA unitary $V \in S_0$ with discretized parameters such that the parameters each differ by less than $\delta/2$, and then the above holds. That is, $S_0$ forms an $\epsilon$-covering of $S$ for $\epsilon = p \delta n^2$.  
\end{proof}

We will need one more preliminary lemma.
\begin{lemma} \label{lemma:subspace_overlaps}
Take any $\epsilon > 2\gamma > 0.$ Let $\{ U_i \}_i$ be a list of unitaries $U_i$ satisfying $\norm{U_i-U_j} \leq \gamma$.  Let $\{ |w_m\rangle \}_{m=1}^M$ be a list of $M$ orthogonal states.  Let $B$ be an orthogonal projection. Consider projectors $B_i = U_i^\dagger B U_i$.  Assume that for all $\{ |w_m\rangle \}_{m=1}^M$ we have \begin{align}
    \langle w_m | B_j | w_m\rangle > \epsilon
\end{align} 
for at least one $B_j$.  Then
 \begin{align}
      M < \frac{\Tr(B)}{\epsilon-2\gamma}.
 \end{align}
\end{lemma}
\begin{proof}
Note that $\norm{B_i-B_j} \leq 2\gamma$. Fix some index $j$. By assumption, for each $|w_m\rangle$ we have 
\begin{align}
\epsilon < \langle w_m | B_i | w_m\rangle
\end{align}
for some index $i$. Then
\begin{align}
\epsilon < \langle w_m | B_i | w_m\rangle = \langle w_m | B_j + (B_i-B_j) | w_m\rangle \leq \langle w_m | B_j  | w_m\rangle + 2\gamma,
\end{align}
and so
\begin{align}
    \langle w_m | B_j | w_m\rangle > \epsilon-2\gamma
\end{align}
holds for all $\{|w_m\rangle\}$. Finally, summing over these states, we get 
\begin{align}
 (\epsilon-2\gamma) M < \sum_{m=1}^M \langle w_m | B_j | w_m\rangle  \leq \Tr(B_j) = \Tr(B) 
\end{align}
and the conclusion follows.
\end{proof}

\subsection{Proof of Theorem
\ref{thm:improvable_main}}
To prove Theorem \ref{thm:improvable_main}, we will focus on the following more basic framing.  Afterward, we will observe how it implies \ref{thm:improvable_main}.
\begin{theorem}\label{thm:improvable}
Consider the depth-$p$ QAOA, and let $M$ be the total number of strings (of any cost) that are improvable to cost $C_1$ with selection probability at least $\Delta$.  Let $d_1$ be the total number of strings with cost at least $C_1$.  Then
\begin{align}
    M \leq \frac{2d_1}{\Delta}\left(\frac{16 \pi p n^2}{\Delta}\right)^{2p}.
\end{align}
\end{theorem}
\begin{proof}
Let $B$ be the orthogonal projection onto the ``best'' strings, specified by all those with cost at least $C_1$.  Let $d_1 = \Tr(B)$ be the number of such strings.  Let $W$  be the set of strings $w$ with the property that there exists some QAOA unitary $U_w$ (depending on $w$) of depth $p$ such that 
\begin{align}
    \langle w |U_w^\dagger B U_w|w\rangle > \Delta.\label{eq:overlap}
\end{align}
Then Theorem \ref{thm:improvable} is precisely the statement that
\begin{align}
    |W| \leq  2\frac{ d_1}{\Delta} |S_0| \leq  2\frac{ d_1}{\Delta} \left(\frac{16 \pi p n^2}{\Delta}\right)^{2p}.
\end{align}
We now argue this inequality. For each $w \in W$, we can assign some $U_w$ such that~\eqref{eq:overlap} holds. Denote the assignment
\begin{align}
    f : & W \to X, \\
     & w \mapsto U_w.
\end{align}
where $X= \{U_w : w \in W\}= f(W)$.  Note that $f$ may be a many-to-one function.

Let $\gamma = \Delta/4.$ Choose a subset $X_1 \subset X$ that is a $\gamma$-net, i.e.\ a $\gamma$-packing and $\gamma$-cover. (See Section \ref{sec:nets}.) That is, the points in $X_1$ are at least $\gamma$ apart, and every point in $X$ is within $\gamma$ of $X_1$. For any $U \in X_1$, consider the $\gamma$-ball $\text{Ball}_\gamma(U) \subset X_1$. Consider the strings $f^{-1}(\text{Ball}_\gamma(U)) \subset W$.  Applying Lemma \ref{lemma:subspace_overlaps} to the unitaries  $\text{Ball}_\gamma(U)$ and the strings $f^{-1}(\text{Ball}_\gamma(U))$, we conclude 
\begin{align}
    |f^{-1}(\text{Ball}_\gamma(U))| \leq \frac{\Tr(B)}{\Delta-2\gamma} = 2\frac{ d_1}{\Delta}.
\end{align}
We apply the above for every $U \in X_1$.  Because $X = \cup_{U \in X_1} \text{Ball}_\gamma(U)$, we have
\begin{align}
    |W| = \left|f^{-1}(X)\right| \leq \left| \bigcup_{U \in X_1} f^{-1}(\text{Ball}_\gamma(U))\right| \leq 2\frac{ d_1}{\Delta} |X_1|.
\end{align}

From Lemma \ref{lem:QAOA_cover} we see that the space $S$ of QAOA unitaries with $2p$ parameters has an $\epsilon$-covering $S_0 \subset S$ for $\epsilon=p \delta n^2$, with 
\begin{align}
    |S_0| = \left(\frac{2\pi}{\delta}\right).
\end{align}
We choose $\delta = \frac{\Delta}{8pn^2}$ so that $S_0$ is an $\epsilon$-covering of $S$ for $\epsilon = \frac{\Delta}{8}$.  Then $X_1$, which is a $(\frac{\Delta}{4})$-packing, can be no larger than the cover $S_0$.  So
\begin{align}
    |W| \leq  2\frac{ d_1}{\Delta} |S_0| \leq  2\frac{ d_1}{\Delta} \left(\frac{16 \pi p n^2}{\Delta}\right)^{2p},
\end{align}
as desired.
\end{proof}
For convenience we will restate\footnote{In contrast to the version stated in the main text, we consider an initial string chosen randomly from within some finite window of costs, rather than at a single fixed cost.  Using a window accommodates cost functions where all strings may have distinct costs.} Theorem \ref{thm:improvable_main} and then see how it follows from the above result, Theorem \ref{thm:improvable}.  

\begin{theorem}
Consider warm-start QAOA at depth $p$, starting with a string of cost within $[C_0, C_0']$ chosen uniformly at random. Let $d_0$ be the number of strings with cost within $[C_0, C_0']$, and let $d_1$ be the number of strings with cost at least $C_1$.  Then with probability at least $1-\epsilon$, the initial string will not be improvable to cost $C_1$ using selection probability at least $\epsilon$, where
\begin{align}
    \epsilon = \left(\frac{2d_1}{d_0}\right)^{\frac{1}{2p+2}}(16 \pi p n^2)^{\frac{p}{p+1}}.
\end{align}
\end{theorem}
\begin{proof}
For convenience let 
\begin{align}
    A = \frac{2d_1}{d_0} (16 \pi p n^2)^{2p}.
\end{align}  Then by the previous theorem, 
\begin{align}
    M \leq d_0 A \Delta^{-(2p+1)}.
\end{align}  
where $M$ is the number of strings improvable to cost $C_1$, with selection probability at least $\Delta$, and within depth at most $p$.
Take 
\begin{align}
    \Delta = A^{\frac{1}{2p+2}}.
\end{align}  Then 
\begin{align}
    \frac{M}{d_0} \leq A \Delta^{-\frac{2p+1}{2p+2}} = A^{\frac{1}{2p+2}} = \Delta.
\end{align}
Note the left-hand side is precisely the fraction of strings within cost $[C_0, C'_0]$ that are improvable in the desired sense.  Therefore we can define $\epsilon = \Delta$ and the corollary holds.
\end{proof}

\end{appendices}

\end{document}